\documentclass[journal,10pt,twocolumn]{IEEEtran}
\usepackage{amsmath,amsfonts}
\usepackage{setspace}
\usepackage{amssymb}
\usepackage{amsthm}
\usepackage{algorithm}
\usepackage{algorithmicx}
\usepackage{algpseudocode}
\usepackage{array}
\usepackage{textcomp}
\usepackage{stfloats}
\usepackage{url}
\usepackage{verbatim}
\usepackage{graphicx}
\usepackage{cite}
\usepackage{dcolumn}
\usepackage{bm}
\usepackage{xcolor}
\usepackage{qcircuit}
\usepackage{braket}
\usepackage{physics}
\usepackage{float}
\usepackage{mathtools}
\usepackage{dirtytalk}
\usepackage{hyperref}
\usepackage{newlfont}
\usepackage{array}
\usepackage{setspace}
\usepackage{subfigure}
\newtheorem{theorem}{Theorem}

\newtheorem{proposition}{Proposition}

\newtheorem{remark}{Remark}

\algnewcommand{\algorithmicand}{\textbf{ and }}
\algnewcommand{\algorithmicor}{\textbf{ or }}
\algnewcommand{\OR}{\algorithmicor}
\algnewcommand{\AND}{\algorithmicand}

\usepackage[T1]{fontenc}

\begin{document}

\title{\huge Reconfigurable Intelligent Surface (RIS)-Assisted Entanglement Distribution in FSO Quantum Networks}

\author{{Mahdi Chehimi, Mohamed Elhattab, Walid Saad, Gayane Vardoyan, Nitish K. Panigrahy, Chadi Assi, and Don Towsley}\vspace{-1cm}
\thanks{Mahdi Chehimi and Walid Saad are with the Bradley Department of Electrical and Computer Engineering, Virginia Tech, USA. e-mail: (mahdic@vt.edu, walids@vt.edu)}
\thanks{Mohamed Elhattab is with the ECE Department, Concordia University, Montreal, QC, Canada and Chadi Assi is with the CIISE Department, Concordia University, Montreal, QC, Canada. email: (m\_elhatt@encs.concordia.ca, chadi.assi@concordia.ca)}
\thanks{Gayane Vardoyan is with QuTech and EEMCS, Delft University of Technology, the Netherlands, and MCICS, University of Massachusetts Amherst, MA, USA. e-mail: gvardoyan@cs.umass.edu.}
\thanks{Nitish K. Panigrahy and Don Towsley are with the Manning College of Information and Computer Sciences, University of Massachusetts Amherst, USA. e-mail: (nitish@cs.umass.edu, towsley@cs.umass.edu)}

\thanks{This research was supported by the U.S. National Science Foundation under Grant CNS-2030215.}}


\markboth{IEEE Transactions on Wireless Communications}%
{Shell \MakeLowercase{\textit{et al.}}: A Sample Article Using IEEEtran.cls for IEEE Journals}



\maketitle

\begin{abstract} 
Quantum networks (QNs) relying on free-space optical (FSO) quantum channels can support quantum applications in environments wherein establishing an optical fiber infrastructure is challenging and costly. However, FSO-based QNs require a clear line-of-sight (LoS) between users, which is challenging due to blockages and natural obstacles. In this paper, a \emph{reconfigurable intelligent surface (RIS)}-assisted FSO-based QN is proposed as a cost-efficient framework providing a virtual LoS between users for entanglement distribution. A novel modeling of the quantum noise and losses experienced by quantum states over FSO channels defined by atmospheric losses, turbulence, and pointing errors is derived. Then, the joint optimization of entanglement distribution and RIS placement problem is formulated, under heterogeneous entanglement rate and fidelity constraints. This problem is solved using a simulated annealing metaheuristic algorithm. Simulation results show that the proposed framework effectively meets the minimum fidelity requirements of all users' quantum applications. This is in stark contrast to baseline algorithms that lead to a drop of at least $83\%$ in users’ end-to-end fidelities. The proposed framework also achieves a $64\%$ enhancement in the fairness level between users compared to baseline rate maximizing frameworks. Finally, the weather conditions, e.g., rain, are observed to have a more significant effect than pointing errors and turbulence.
\end{abstract}

\begin{IEEEkeywords}
quantum networks, reflective intelligent surfaces, free-space quantum communication 

\end{IEEEkeywords}

\vspace{-0.45cm}
\section{\label{sec:Indroduction}Introduction}\vspace{-0.1cm}
\IEEEPARstart{Q}{uantum} networks (QNs) will facilitate the distribution of quantum states, or qubits, across distant nodes \cite{chehimi2023roadmap}, thereby unlocking novel applications in quantum sensing \cite{degen2017quantum}, distributed quantum computing \cite{chehimi2023foundations}, and quantum communication protocols such as quantum key distribution (QKD) \cite{chehimi2022physics}. However, practically deploying and operating QNs is nontrivial due to their unique underlying physics, hardware operations, and probabilistic nature compared to classical networks \cite{chehimi2022physics}.

For instance, state-of-the-art QN architectures primarily rely on optical fiber channels for communicating and sharing quantum information. However, in areas with obstacles, natural barriers, or limited infrastructure, setting up a fiber infrastructure can be difficult. In such cases, \emph{free-space optical (FSO)} quantum channels offer a promising alternative. In an FSO channel, quantum optical signals are wirelessly sent through the air between terrestrial stations or through outer space using satellites \cite{pirandola2021limits}. Thus, FSO-based QNs are more flexible and cost-efficient compared to fiber-based QNs. Other advantages of an FSO-based QN include secure connections due to the difficulty of intercepting a light beam without leaving a trace, and high transmission rates by utilizing a wide range of wavelengths and frequencies \cite{liao2017long}.


Despite these advantages, designing an FSO-based QN requires overcoming several major challenges. First, a common feature found in QNs is having a star-shaped architecture with a central \emph{quantum base station (QBS)} that shares \emph{entangled qubits} with connected users. The QBS has limited capacity to generate entangled pairs of qubits, which are fragile resources that have a short coherence time, and this makes the problem of allocating these resources to QN users challenging \cite{chehimi2021entanglement_rate_optimization}. Particularly, entangled qubit allocation entails controlling single-photon sources and quantum memories to ensure that quantum states are generated at sufficient rate and quality, or \emph{fidelity}, to satisfy heterogeneous user requirements. In addition, the most substantial challenge faced by an FSO-based QN is the need for a clear line-of-sight (LoS) between the communicating nodes. Without LoS, due to the highly-focused and directional nature of optical signals, it is impossible to transmit quantum optical signals over an FSO channel between distant nodes \cite{pirandola2021limits}. Indeed, blockages and obstacles can obstruct an FSO quantum link by breaking the LoS connection between QN nodes \cite{alshaer2021hybrid}. Furthermore, establishing and maintaining an LoS in an FSO-based QN becomes more challenging when considering practical environmental effects such as atmospheric losses, turbulence, misalignment, and pointing errors between communicating nodes \cite{hassan2023experimental}. These challenges can result in quantum signal loss and quality degradation. As a result, two research questions arise: \begin{enumerate}
    \item \emph{How do we establish and maintain an LoS connection between distant nodes communicating over an FSO channel in a multi-user QN where blockages and atmospheric effects are present?}

    \item \emph{How do we efficiently allocate available QBS entangled resources to different users so as to fairly maximize their end-to-end (E2E) service rates while satisfying their quality-of-service (QoS) requirements?} 
\end{enumerate} 
\vspace{-0.4cm}
\subsection{Prior Works}
FSO-based QNs have received considerable attention \cite{liao2017long,ntanos2021availability,alshaer2021hybrid,hassan2023experimental} due to their advantages over fiber-based QNs. For instance, the authors in \cite{liao2017long} reported a practical experiment for long distance quantum communications over FSO channels while overcoming weather-dependent challenges and losses. Additionally, the work in \cite{alshaer2021hybrid} considered hybrid MPPM-BB84 QKD deployment over an FSO channel while considering atmospheric turbulence and pointing errors. Moreover, the impacts of atmospheric turbulence on BB84 QKD over FSO quantum channels were discussed in \cite{hassan2023experimental}. Finally, the work in \cite{ntanos2021availability} considered the utilization of decoy quantum states to execute the BB84 QKD protocol over a terrestrial FSO quantum channel while considering several environmental impacts, like atmospheric turbulence, scattering, and rain attenuation. However, all these works, \cite{liao2017long,ntanos2021availability,alshaer2021hybrid,hassan2023experimental}, considered FSO-based QNs with no natural or human-made obstructions that could block the LoS connection, and none of them considered practical scenarios in which LoS is missing in the QN. 


For FSO-based QN scenarios in which a direct LoS link is unavailable, it is possible to establish a \emph{virtual LoS} between the transmitter and receiver nodes. Such a virtual LoS is established by identifying an intermediate point that has a clear LoS with both the transmitter and receiver. Signals are then transmitted from the transmitter node and reflected off this intermediate point, ensuring they reach the receiver node. One way to practically realize such a virtual LoS is through the use of quantum repeaters at the intermediate nodes \cite{pirandola2021limits}. However, such quantum repeaters require the inclusion of a partial or complete quantum FSO transceiver chain, including: entanglement generation sources, entanglement swapping modules, single-photon detectors, quantum memories, signal processing units, and tracking systems \cite{9443170}. These requirements impose significant design complexities, thereby making the use of quantum repeaters an expensive solution. Another strategy involves integrating satellites as relay nodes or optical reflectors within FSO-based QNs with obstructions \cite{hosseinidehaj2018satellite}. While this offers a potential practical solution, it brings its own set of challenges. Primarily, this method substantially elevates operational costs. Additionally, there is an inherent increase in communication delays due to the extended distances the photons must traverse. This greater travel distance necessitates advancements in photon emission technologies to generate higher quality photons, in addition to quantum memories with longer lifetimes.

One promising solution to overcome the LoS challenges of FSO-based QNs in an efficient and cost-effective manner could be through the concept of a \emph{reconfigurable intelligent surface (RIS)}. Indeed, prior works in the field of RIS for classical FSO communications \cite{10130384,9466323,9443170} already demonstrated the potential of using an optical passive RIS to alleviate the LoS requirement for an FSO system and to establish a virtual LoS between communicating nodes. RISs offer notable advantages in terms of energy efficiency and cost-effectiveness, and they do not require a quantum memory. They are comprised of passive elements and can be conveniently deployed on existing infrastructure such as lamp posts and building facades \cite{10130384,9466323,9443170}. Although there is an inherent compatibility between RISs and optical quantum signals, to the best of our knowledge, only two prior works have considered the presence of an RIS in a QN \cite{kundu2023intelligent} and \cite{kisseleff2023trusted}. Nonetheless, the work in \cite{kundu2023intelligent} does not take into account real-world quantum considerations, and the work in \cite{kisseleff2023trusted} is focused on classical communications and makes a strong assumption that the generation of QKD secret keys is always successful. In particular, the work in \cite{kundu2023intelligent} incorporates an RIS in a single point-to-point FSO quantum link without considering an FSO-based QN where an RIS is optimally placed between multiple users. Moreover, the model in \cite{kundu2023intelligent} did not allocate EGR and quantum memory capacity to QN users to satisfy their unique QoS requirements. Additionally, the authors in \cite{kundu2023intelligent} did not analyze how environmental noise affects the \emph{fidelity} of quantum states over FSO links. On the other hand, the authors in \cite{kisseleff2023trusted} considered a multi-user network secured by QKD. However, QKD processes in \cite{kisseleff2023trusted} are assumed to successfully distribute secret keys without any consideration of the limitations associated with the entanglement generation, storage, or distribution operations. Moreover, the impact of environmental effects on quantum signals both in terms of the resulting losses and noise, corresponding to rate and fidelity, is missing in \cite{kisseleff2023trusted}. In general, the works in \cite{kundu2023intelligent} and \cite{kisseleff2023trusted} distribute entangled qubits in FSO-based QNs using classical approaches that are agnostic to quantum state fidelity. This is an impractical entanglement distribution approach since each quantum application, e.g., QKD, can only tolerate a certain amount of error, which can be translated into a minimum average fidelity requirement. These traits lead to quantum-specific effects in FSO signals, particularly in presence of an RIS \cite{chehimi2022physics}. While practical experiments validated the physical realization of FSO-based QNs with reflectors \cite{liao2017long}, \emph{no prior work has explored the integration of an RIS into a multi-user FSO-based QN with diverse QoS requirements while analyzing the resulting quantum noise and loss effects on quantum state fidelity and rate. Additionally, there is a lack of research on allocating EGR in FSO-based QNs with RISs while accounting for environmental factors such as attenuation, turbulence, and pointing errors.}




\vspace{-0.3cm}
\subsection{Contributions}
The main contribution of this paper is a novel RIS-assisted FSO-based terrestrial QN architecture designed to facilitate fair EGR allocation among users with heterogeneous QoS requirements, taking into account practical environmental effects, obstructions, and blockages. In particular, we propose the first model that jointly accounts for the losses and noise in FSO quantum channels due to real-world environmental conditions, analyzing their effects on the rate and fidelity of quantum signals. Towards this goal, we make key contributions:
\begin{itemize}
\item We introduce a novel model of a star-shaped, FSO-based terrestrial QN in the presence of environmental effects, taking into account practical challenges imposed by blockages. The proposed model is the first in the literature to perform quantum EGR allocation while addressing the problem of absent LoS in FSO-based terrestrial QNs.

\item We conduct a comprehensive analysis of the passive RIS integration into FSO-based terrestrial QN architectures. Specifically, we analyze the impact of various environmental factors, including atmospheric loss, turbulence, and pointing errors, on the generation, preservation, distribution, and fidelity of entangled quantum states. To do so, we derive a closed-form expression for the probability of successfully sending an entangled photon over an FSO quantum channel with various environmental effects. We then propose a model for quantum phase noise induced by atmospheric turbulence on entangled photons over an FSO quantum channel, where the quantum noise directly depends on the turbulence strength. This is the first analysis that captures both the cumulative losses and quantum-specific noise encountered in practical FSO quantum channels under environmental effects.

\item We formulate an optimization problem to jointly optimize the RIS placement and QN EGR allocation while guaranteeing fairness among the end users and considering heterogeneous QoS requirements on minimum rate and fidelity for different quantum applications. Then, we propose a metaheuristic simulated annealing algorithm to solve the formulated problem efficiently. 
\end{itemize}

Simulation results show that the proposed framework outperforms the classical, fidelity-agnostic resource allocation frameworks in terms of satisfying minimum fidelity requirements. In contrast to classical fidelity-agnostic frameworks whose users' fidelity remains at least 83\% below the required minimum fidelity, our framework manages to satisfy the fidelity requirements for all users' quantum applications. Additionally, our framework achieves a 64\% enhancement in the fairness level between users compared to rate maximizing frameworks that do not incorporate fairness constraints in their analysis. Furthermore, we show that the E2E distance of QN users has a more severe impact on their E2E rate compared to their quantum application's moderate minimum fidelity requirement. We verify that the proposed framework can serve a high number of users while satisfying their heterogeneous fairness, rate, and fidelity requirements. Finally, we show that weather conditions have the strongest impact among the different environmental effects as changing the weather from being sunny to being rainy could result in around 45\% reduction in the E2E sum rate.

The rest of this paper is organized as follows. Section \ref{section_system_model} describes the system model, its constituents along with their geometry, and the quantum FSO channel model. Then, in Section \ref{sec_III}, we model both losses and noise encountered in the system model, and analyze their effect on rate and fidelity of entangled states. Section \ref{sec_optimization} introduces our proposed optimization framework, while Section \ref{sec_conclusion} draws conclusions.
\vspace{-0.45cm}
\section{System Model}\label{section_system_model}
\begin{figure}[t!]
\centering
\includegraphics[width=\columnwidth]{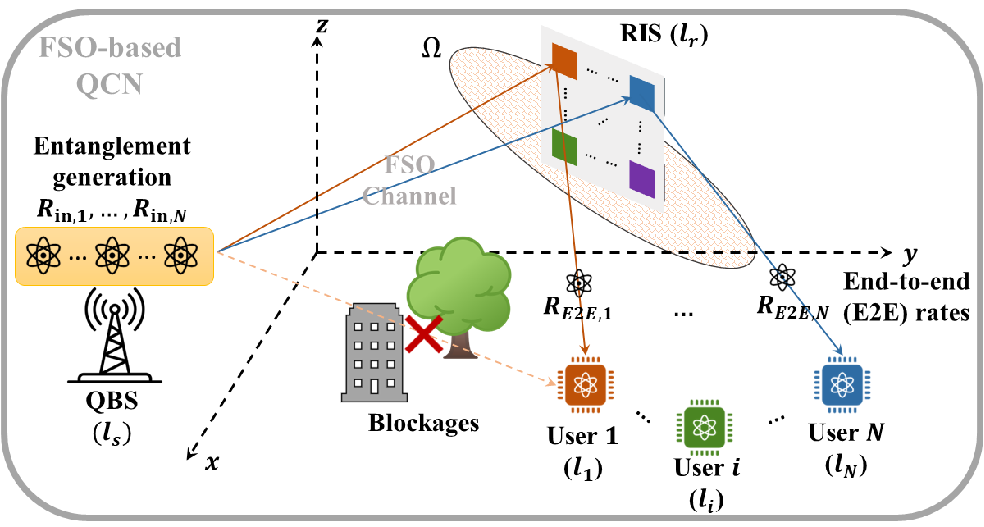}\vspace{-0.1cm}
\caption{Studied system model of an RIS-assisted FSO-based QN with blockages that result in the absence of direct LoS connections between the QBS and end users.}
\label{fig:system_model}\vspace{-0.45cm}
\end{figure}

\subsection{Quantum Network (QN) Elements and Geometry}
As shown in Fig. \ref{fig:system_model}, we consider a star-shaped terrestrial QN comprised of a QBS connected to a set $\mathcal{N}$ of $N$ end users through FSO quantum channels. We assume that end users fall into the dark zone of the QBS, i.e., the direct LoS FSO quantum links between the QBS and end users are blocked or weak due to extreme conditions, such as obstruction by high rise buildings.\footnote{This is a challenging scenario for traditional terrestrial QNs without RIS or satellite support, and it is a typical assumption in the literature of RIS-assisted communication scenarios \cite{10130384,9466323,9443170}.} To circumvent this challenge, an RIS is deployed to provide an alternative FSO quantum link or a virtual LoS between end users and QBS.



In the considered area, the QBS is located at the origin $\textit{\textbf{l}}_s$ = ($x_s$,$y_s$,$H_s$), the RIS is located at $\textit{\textbf{l}}_r$ = ($x_r$,$y_r$,$H_r$), while each end user $i\in\mathcal{N}$, is located at $\textit{\textbf{l}}_i$ = ($x_i$,$y_i$,$H_i$). The QBS is considered to be the source that generates \emph{entangled pairs of qubits (or photons)}, which are shared with the end users to create quantum communication links and enable applications, like QKD. The RIS must be placed so as to optimize the distribution of entangled qubits from the QBS to the users. This is essential to ensure that the quantum applications meet their QoS requirements in terms of the achieved rate and fidelity. Next, we introduce $\Omega$ as the permitted region for deploying the RIS where $\Omega$ satisfies the following condition:
\begin{equation}\label{eq_feasible_RIS_region}
\begin{split}
    \Omega &= \{(x_r, y_r, H_r) \mid
    x_{\min} \leq x_r \leq x_{\max}, \\
    &y_{\min} \leq y_r \leq y_{\max},
    H_{\min} \leq H_r \leq H_{\max}\},
\end{split}
\end{equation}
where $\left[x_{\min}, x_{\max}\right], \left[y_{\min}, y_{\max}\right]$ and $\left[H_{\min}, H_{\max}\right]$ represent the possible deployment ranges along the x-, y-, and z-axes, respectively. The size of $\Omega$, in practice, is usually constrained due to infrastructural considerations, including the availability of appropriate building facades for RIS integration, the availability of LoS links with the QBS and the QN users at available infrastructure, and the distribution of surrounding obstacles in deployment locations, all of which constrain the deployment options. The distance between the QBS and the RIS is $d_{s,r}(\textit{\textbf{l}}_r) = \sqrt{(x_r-x_s)^2+(y_r-y_s)^2+(H_r-H_s)^2}$, and the distance from the RIS to each user $i\in\mathcal{N}$ is $d_{r,i}(\textit{\textbf{l}}_r) = \sqrt{(x_i-x_r)^2+(y_i-y_r)^2+(H_i-H_r)^2}$. Accordingly, the \emph{end-to-end (E2E)} distance between the QBS and a user $i\in\mathcal{N}$ is $d_{\text{E2E},i}(\textit{\textbf{l}}_r) = d_{s,r}(\textit{\textbf{l}}_r) + d_{r,i}(\textit{\textbf{l}}_r)$.

The RIS can dynamically adjust the phase shifts of its individual elements to create a virtual LoS and to direct photons sent from the QBS toward the different end users. For simultaneously serving different users, the RIS is partitioned into $N$ sub-surfaces, each of which consists of fixed units of RIS elements and is dedicated to a particular user $i\in\mathcal{N}$ (e.g., see \cite{9976948}). The phase shift of each sub-surface is adjusted according to the location of its corresponding end user.\footnote{This can be achieved by adopting one of the existing techniques for RIS phase control like the one proposed in \cite{9443170}.} This adjustment ensures that the entangled photons generated and transferred from the QBS are appropriately directed towards their intended user $i\in\mathcal{N}$, as shown in Fig. \ref{fig:system_model}. The relative distances between the RIS and both QBS and end users are significantly larger than the distances between the RIS partitions. Henceforth, the QN operation is performed at the far-field regime of the RIS, which is a common assumption \cite{9976948}. Hence, it is reasonable to represent the location of the RIS by a single point $\textit{\textbf{l}}_r$ = ($x_r$,$y_r$,$H_r$).

\vspace*{-0.4cm}
\subsection{Quantum Channel Model}\label{sec_quantum_channel_Model}

The E2E FSO channel gain $h_i,~\forall i\in \mathcal{N}$ between the QBS and end user $i\in\mathcal{N}$, is given by \cite{9443170}:

\vspace*{-0.3cm}
\begin{equation}\label{eq_channel}
    h_i(\textit{\textbf{l}}_r) = \varsigma \eta h_{i}^\text{p}(\textit{\textbf{l}}_r) h_{i}^\text{a}(\textit{\textbf{l}}_r) h_{i}^\text{g}(\textit{\textbf{l}}_r), \quad \forall i \in \mathcal{N},
\end{equation}
where $\varsigma$ is the RIS reflection efficiency, i.e., the fraction of power reflected by the RIS, which depends on both the bias voltage applied to the RIS and the operating frequency, $\eta$ is the responsivity of the receiver, and $h_{i}^\text{p}, h_{i}^\text{a}$, and $h_{i}^\text{g}$ are, respectively, the atmospheric loss, atmospheric turbulence, and pointing error loss.

\subsubsection{Atmospheric loss} $h_{i}^\text{p}$ is a deterministic loss and represents the reduction of power along a transmission path caused by the scattering and absorption of light by particles in the atmosphere. Accordingly, the atmospheric loss is \cite{jamali2016link}:
\begin{equation}
\begin{split}
    h_{i}^\text{p}(\textit{\textbf{l}}_r) &= 10^{-\hat{\kappa} d_{s,r}(\textit{\textbf{l}}_r)/10}\times 10^{-\hat{\kappa} d_{r,i}(\textit{\textbf{l}}_r)/10} \\
    &= 10^{-\hat{\kappa} d_{\mathrm{E2E},i}(\textit{\textbf{l}}_r)/10},\quad\forall i \in\mathcal{N},
\end{split}
\label{pathloss}
\end{equation}
where $\hat{\kappa}$ is the weather-dependent attenuation coefficient of the FSO link. From \eqref{pathloss}, we observe that the atmospheric loss $h_{i}^\text{p}$ consists of a product of two terms, namely the atmospheric loss of the QBS-RIS link $10^{-\hat{\kappa }d_{s,r}(\textit{\textbf{l}}_r)/10}$ and the atmospheric loss of the RIS-user link $10^{-\hat{\kappa} d_{r,i}(\textit{\textbf{l}}_r)/10}$.

\subsubsection{Atmospheric turbulence} Under moderate-to-strong atmospheric turbulence, $h_{i}^\text{a}$ can be modeled by a Gamma-Gamma distribution, and, hence, the probability density function (PDF) of $h_{i}^\text{a}$ is \cite{khalighi2014survey}:\vspace{-0.4cm}

\begin{equation}\small
\begin{split}
    f_{h_{i}^\text{a}}(h_{i}^\text{a}) &= \frac{2(\alpha_i\beta_i)^{(\alpha_i+\beta_i)/2}}{\Gamma(\alpha_i)\gamma(\beta_i)}\times (h_{i}^\text{a})^{(\alpha_i+\beta_i)/2-1} \\
    &\times K_{\alpha_i-\beta_i}(2\sqrt{\alpha_i\beta_i h_{i}^\text{a}}),
\end{split}
\label{TurbPDF}
\end{equation}
where $\Gamma(\cdot)$ is the Gamma function and $K_\gamma$ is the $\gamma$-order modified Bessel function of the second kind. $\beta_i$ and $\alpha_i$ represent, respectively, the effective number of large-scale and small-scale cells of the scattering process that depend on the RIS location, and are given by \cite{khalighi2014survey}:\vspace{-0.4cm}


\begin{equation}\small
\begin{split}
\alpha_i(\mathbf{l}_r) &= \left[ \text{exp}\left( \frac{0.49\sigma_{R,i}^2(\mathbf{l}_r)}{\left(1+1.11\sigma_{R,i}^{\frac{12}{5}}(\mathbf{l}_r)\right)^{\frac{7}{6}}}\right)-1\right]^{-1},
\end{split}
\label{eq:alpha}
\end{equation}

\begin{equation}\small
\begin{split}
\beta_i(\mathbf{l}_r) &= \left[ \text{exp}\left( \frac{0.51\sigma_{R,i}^2(\mathbf{l}_r)}{\left(1+0.69\sigma_{R,i}^{\frac{12}{5}}(\mathbf{l}_r)\right)^{\frac{5}{6}}}\right)-1\right]^{-1},
\end{split}
\label{eq:beta}
\end{equation}
where $\sigma_{R,i}^2(\textit{\textbf{l}}_r) = 1.23C_n^2k^{7/6}d_{\text{E2E},i}^{11/6}(\textit{\textbf{l}}_r)$ is the \emph{Rytov variance} which captures phase fluctuations due to turbulence. Moreover, $k = 2\pi/\lambda$ is the optical wave number, $\lambda$ is the wavelength, and $C_n^2$ is the refractive index structure constant, which quantifies the turbulence strength.

\subsubsection{Pointing Error Model} Pointing errors in RIS-assisted free-space terrestrial QNs stem mainly from jitter at the transmitter and at the RIS. Therefore, the fading due to pointing errors at user $i\in\mathcal{N}$ will be given by \cite{10130384,9466323}:
\begin{equation}\footnotesize
    h_{i}^\text{g}(\textit{\textbf{l}}_r) \approx A_{0,i}(\textit{\textbf{l}}_r) \exp\left(- 2 \varrho_i^2(\textit{\textbf{l}}_r)/W^2_{z_{\mathrm{eq},i}}(\textit{\textbf{l}}_r)\right), \quad \forall i \in \mathcal{N}, 
\end{equation}
where $A_{0,i}(\textit{\textbf{l}}_r) = \left[\mathrm{erf}\left(v_i(\textit{\textbf{l}}_r)\right)\right]^2, W^2_{z_{\mathrm{eq},i}}(\textit{\textbf{l}}_r) = \frac{W_{z,i}^2(\textit{\textbf{l}}_r)\sqrt{\pi}\mathrm{erf}\left(v_i(\textit{\textbf{l}}_r)\right)}{2v_i(\textit{\textbf{l}}_r)\exp\left(-v_i^2(\textit{\textbf{l}}_r)\right)},$ and $v_i(\textit{\textbf{l}}_r) = \frac{\sqrt{\pi}r_a}{\sqrt{2}W_{z,i}(\textit{\textbf{l}}_r)}$. Here, $r_a$ is the aperture radius, $W_{z,i}(\textit{\textbf{l}}_r) = \phi_d d_{\textrm{E2E},i}(\textit{\textbf{l}}_r)$ represents the beam width, $\phi_d$ is the beam divergence angle, which are related to the spreading of the photon's wave function. $\varrho_i(\textit{\textbf{l}}_r) = \mathrm{tan}(\theta_{s,i}(\textit{\textbf{l}}_r))d_{r,i}(\textit{\textbf{l}}_r) \approx \theta_{s,i}(\textit{\textbf{l}}_r)d_{r,i}(\textit{\textbf{l}}_r)$ is the instantaneous displacement between the actual receiving point of the beam and the center of the receiver, where $\theta_{s,i}$ is the superimposed pointing error angle constructed by the actual incident point at user $i$, the reflection point at the RIS, and the center of receiver $i$. $\theta_{s,i}(\textit{\textbf{l}}_r) = \sqrt{\theta_{sx,i}^2(\textit{\textbf{l}}_r) + \theta_{sy,i}^2(\textit{\textbf{l}}_r)}$,  where $\theta_{sx,i}(\textit{\textbf{l}}_r) \approx \theta_{x,i} \left(1 + \frac{d_{s,r}(\textit{\textbf{l}}_r)}{d_{r,i}(\textit{\textbf{l}}_r)} + 2\phi_{x,i}\right)$ is the horizontal component of $\theta_{s,i}$ and $\theta_{sy,i}(\textit{\textbf{l}}_r) \approx \theta_{y,i} \left(1 + \frac{d_{s,r}(\textit{\textbf{l}}_r)}{d_{r,i}(\textit{\textbf{l}}_r)} + 2\phi_{y,i}\right)$ is the vertical component of $\theta_{s,i}$. Here, $\theta_{x,i}$ and $\theta_{y,i}$ are random variables that follow a Gaussian distribution with zero mean and variance $\sigma_{\theta}^2$ \cite{9466323}. $\phi_{y,i}$ and $\phi_{x,i}$ are the deflection angles in vertical and  horizontal planes, respectively, and are modeled as Gaussian distributed random variables with zero mean and variance $\sigma_{\phi}^2$. Hence, the PDF of $\theta_{s,i}$ will be \cite{9466323}: \vspace{-0.2cm}
\begin{equation}\small
\begin{split}
    f_{\theta_{s,i}}(\theta_{s,i}) &= \frac{\theta_{s,i}}{\sigma_{\theta}^2\left(1 + \frac{d_{s,r}}{d_{r,i}}\right)^2 + 4\sigma^2_{\phi}} \\
    &\quad \times \exp\left(-\frac{\theta_{s,i}^2}{2\sigma_{\theta}^2\left(1 + \frac{d_{s,r}}{d_{r,i}}\right)^2 + 8\sigma^2_{\phi}}\right).
\end{split}
\end{equation}
\par By applying random variable transformation, the PDF for the pointing errors can be expressed as follows \cite{9466323}:\vspace{-0.2cm}
\begin{equation}\small
    f_{h_{i}^\text{g}}(h_{i}^\text{g}) = \frac{\vartheta_i}{A_{0,i}^{\vartheta_i}}(h_{i}^\text{g})^{\vartheta_i - 1}, \quad \forall i \in \mathcal{N}, 0 \leq h_{i}^\text{g} \leq A_{0,i}, \label{PointingPDF}
\end{equation}
where $\vartheta_i(\textit{\textbf{l}}_r) = \frac{W_{z_{\mathrm{eq}},i}^2(\textit{\textbf{l}}_r)}{4d_{\text{E2E},i}^2(\textit{\textbf{l}}_r)\sigma_{\theta}^2 + 16 d_{r,i}^2(\textit{\textbf{l}}_r) \sigma_{\phi}^2}$.

The FSO channel model in \eqref{eq_channel} encompasses a wide range of practical imperfections and environmental factors—including atmospheric loss, turbulence, and pointing error. These factors influence the behavior of an entangled photon as it traverses an FSO quantum channel and lead to two challenges: 1) \emph{losses}, that affect the \emph{entanglement generation rate (EGR)}, since some transferred photons may get lost during their transmission, and 2) \emph{noise}, that affects the quality, or \emph{fidelity}, of the E2E entangled states due to qubit interactions with their surroundings. To ensure optimal entanglement distribution, it is crucial to comprehensively analyze and model how the different environmental effects influence E2E rate and fidelity. This analysis, presented in the next section, paves the way for efficient resource allocation strategies and aids in identifying the optimal RIS placement location within FSO-based QNs. 

\vspace{-0.2cm}
\section{Entanglement Generation and Distribution}\label{sec_III}
In a star-shaped FSO-based QN, a QBS is equipped with entanglement generation units that produce pairs of entangled qubits, in addition to a \emph{quantum memory} with limited capacity that can store a maximum rate $C_{\text{max}}$ of entangled qubits per second. Each generated entangled pair of qubits consists of: 1) a \emph{matter qubit}, stored in the quantum memory, and 2) a \emph{flying qubit}, or a single photon of light, sent over the FSO quantum channel to the RIS, where it is then reflected towards its designated user.\footnote{We do not presume a specific hardware technology for the entanglement generation process to keep our network model abstract.}

The matter qubit is preserved during the \emph{coherence time} of the quantum memory. Thus, instead of suffering from strong losses that could result in losing matter qubits in the memory, the main challenge in preserving such qubits is quantum noise due to interactions with the quantum memory. Meanwhile, the flying qubit interacts with the FSO channel and suffers from the imperfections and environmental effects, as discussed in Section \ref{sec_quantum_channel_Model}. Accordingly, the flying qubit significantly suffers from both noise and losses. As discussed earlier, losses have a direct effect on the EGR in the QN, while noise impacts the quality of the distributed entangled states. Next, we analyze each one of these effects separately.

\vspace{-0.1cm}
\subsection{Losses and EGR Analysis}\label{sec_losses_EGR}
The role of the QBS is to generate and fairly distribute entangled qubits to end nodes (referred to interchangeably as \say{users} throughout this paper) with the assistance of an RIS to alleviate the absence of LoS to these users. To do so, the QBS generates pairs of entangled qubits dedicated to each user $i\in\mathcal{N}$ at a rate $R_{\text{in},i}$ pairs per second. Since flying qubits suffer several environmental effects and losses, the rate of entangled qubits successfully delivered to the end nodes will be less than the initial EGR. To quantify such losses, we introduce the \emph{probability of success}, $P_{\mathrm{succ},i}$, as a quantitative measure of the probability of successfully sending an entangled photon from the QBS towards user $i\in\mathcal{N}$. This probability of success incorporates the three discussed sources of photon-loss in the FSO quantum channel (i.e., atmospheric loss, turbulence, and pointing error), and it quantifies the probability of the combined channel gain being larger than a predefined threshold $\zeta_{\mathrm{th}}$. Here, the location of the RIS, $\textit{\textbf{l}}_r$, in the QN plays a fundamental role since it directly affects the distance traveled by photons and significantly impacts the environmental parameters. The closed-form expression for $P_{\mathrm{succ},i}$ is derived next.
\begin{theorem}\label{thrm_prob_success}
    The probability of successfully sending a single entangled photon from the QBS to user $i\in\mathcal{N}$ over an FSO quantum channel (given in \eqref{eq_channel}) considering the atmospheric loss, turbulence, and pointing error is given by: 
    \begin{equation}\footnotesize
    \begin{split}
        &P_{\mathrm{succ},i}(\textit{\textbf{l}}_r) = 1-\Bigg(\frac{\vartheta_i(\textit{\textbf{l}}_r)}{\Gamma(\alpha_i(\textit{\textbf{l}}_r))\Gamma(\beta_i(\textit{\textbf{l}}_r))} \\
        & \times G_{2,4}^{3,1} \left[\frac{\alpha_i(\textit{\textbf{l}}_r)\beta_i(\textit{\textbf{l}}_r) \chi_{\mathrm{th}}}{A_{0,i}(\textit{\textbf{l}}_r)h_{i}^\text{p}(\textit{\textbf{l}}_r)} {\Bigg\vert} \begin{matrix} 
            \multicolumn{2}{c}{1,} & \multicolumn{2}{c}{\vartheta_i(\textit{\textbf{l}}_r) +1} \\ 
            \vartheta_i(\textit{\textbf{l}}_r), & \alpha_i(\textit{\textbf{l}}_r), & \beta_i(\textit{\textbf{l}}_r),&0
            \end{matrix}\right]\Bigg),
    \end{split}
    \end{equation}
where $G_{n,m}^{p,q}$ is the Meijer's G-function, $\Gamma(.)$ is the Gamma function, and $\chi_{\mathrm{th}} = \frac{\zeta_{\mathrm{th}}}{\varsigma \eta}$.
 \end{theorem}
\begin{proof}
    See Appendix \ref{appendix_proof_Theorem_1}.
\end{proof}

From Theorem \ref{thrm_prob_success}, we observe that the probability of success significantly depends on the different environmental parameters along with the Meijer G-function (a general function encompassing several special cases). To further analyze the obtained expression, we run simulations, in Fig. \ref{fig_theorem1}, to understand the performance of the Meijer function, and accordingly, the behavior of the probability of success, as different environmental parameters vary. It is evident from Fig. \ref{fig_theorem1} that the probability of success is significantly reduced in harsher weather conditions and under increased turbulence.

The QBS generates entangled qubits at a rate $R_{\mathrm{in},i},\forall i \in\mathcal{N}$, which are then transferred towards their intended users over an FSO quantum channel characterized by the aforementioned losses. Accordingly, the resulting E2E rate of successfully established E2E links, $R_{\text{E2E},i}$, for user $i\in\mathcal{N}$, can be defined as:
\begin{equation}
    R_{\mathrm{E2E},i}(\textit{\textbf{l}}_r,R_{\mathrm{in},i}) = P_{\mathrm{succ},i}(\textit{\textbf{l}}_r) \cdot R_{\mathrm{in},i}.
\end{equation}
Here, the E2E rate directly depends on the initial EGR and the location of the RIS which significantly affects the probability of successfully sending photons from the QBS toward the end nodes.
\begin{figure}[!t]\vspace{-0.2cm} 
\centering
{\centering\includegraphics[width=0.41\textwidth]{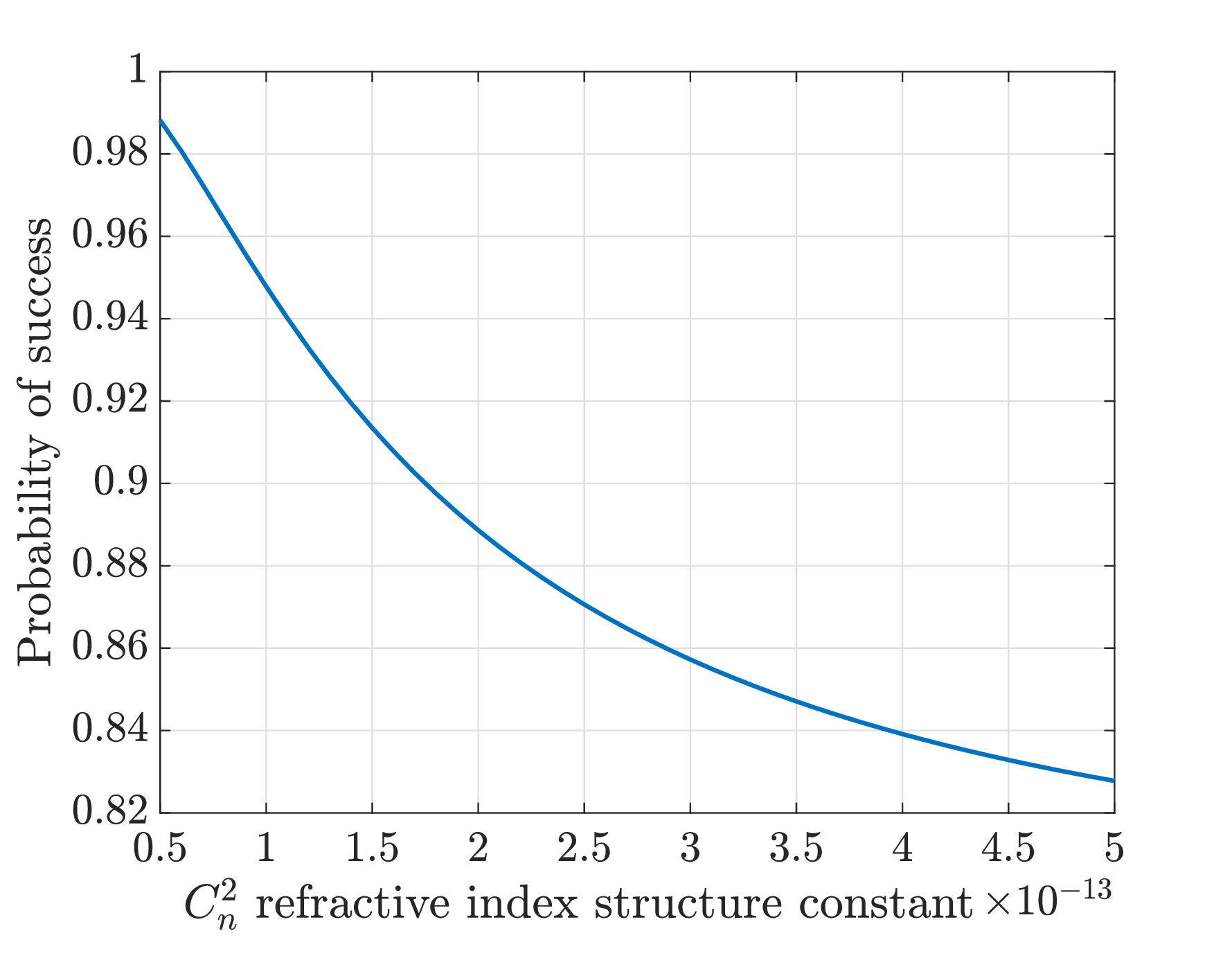}}
{\centering\includegraphics[width=0.41\textwidth]{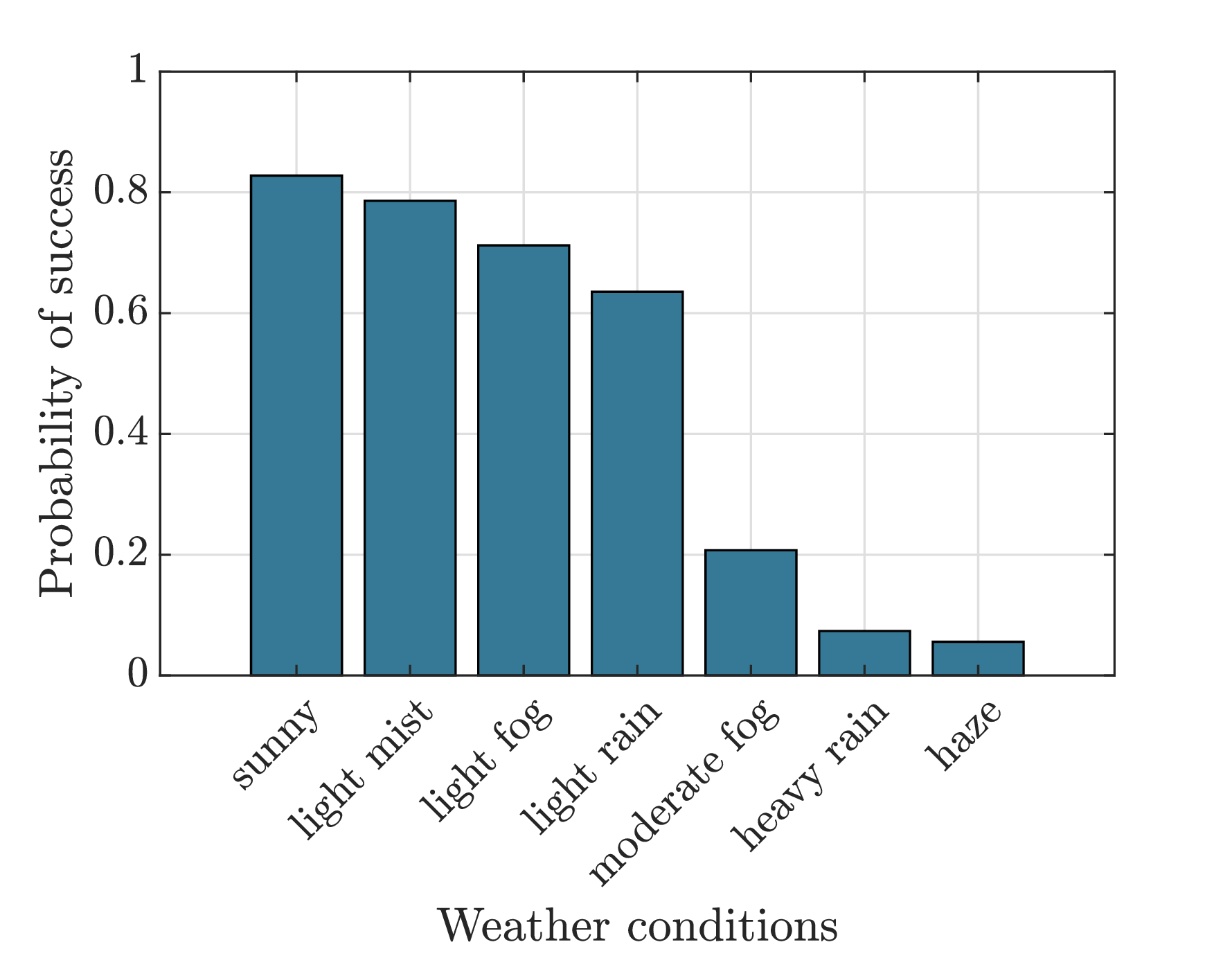}} \vspace{-0.2cm}
\caption{Performance of the proposed probability of success expression in Theorem 1 against turbulence parameter and weather condition.}\label{fig_theorem1}
\vspace{-0.5cm}
\end{figure}

\vspace{-0.2cm}
\subsection{Noise and E2E Entangled State Analysis}\label{sec_quantum_noise}
Each entangled pair of qubits generated by the QBS and dedicated to a user $i\in\mathcal{N}$, is represented by a general Bell-diagonal quantum state, and is given by \cite{renner2008security}:
\begin{equation}\small
    \boldsymbol{\rho}_{\mathrm{BD},i} = \lambda_{00,i}\boldsymbol{\Phi}_{00} + \lambda_{01,i}\boldsymbol{\Phi}_{01} + \lambda_{10,i}\boldsymbol{\Phi}_{10} + \lambda_{11,i}\boldsymbol{\Phi}_{11},
\end{equation}
where $\boldsymbol{\Phi}_{mn} = \ketbra{\Phi_{mn}}{\Phi_{mn}}$, and $\ket{\Phi_{mn}} = {\boldsymbol{\sigma}^m_X}{\boldsymbol{\sigma}^n_Z}\otimes \boldsymbol{I}\ket{\Phi^+}$. Here, $\boldsymbol{\sigma}_X$ and $\boldsymbol{\sigma}_Z$ are the Pauli X- and Z-gates \cite{renner2008security}, $\boldsymbol{I}$ is the identity operator, $m,n\in\{0,1\}$, and $\ket{\Phi^+} = \frac{1}{\sqrt{2}}(\ket{00} + \ket{11})$. Moreover, $\lambda_{00,i} + \lambda_{01,i} + \lambda_{10,i} + \lambda_{11,i} = 1$, and the fidelity of such Bell-diagonal states, with respect to $\Phi_{00}$, is $\bra{\Phi_{00}}\boldsymbol{\rho}_{\mathrm{BD},i}\ket{\Phi_{00}} = \lambda_{00,i}$.

In order for an entangled qubit, successfully distributed from the QBS towards its intended user, to be useful, it must be maintained at a high quality, or \emph{fidelity}. In general, the quality of entangled pairs of qubits is mainly affected by noise that each qubit from the pair experiences. Hence, to quantify the overall quality of distributed E2E entangled pairs, we must model and analyze the noise experienced by both the \emph{matter qubit}, during its storage in the quantum memory of the QBS, and the \emph{flying qubit}, during its transmission over an FSO quantum channel with environmental effects. Analysis of both noise channels allows us to derive the resulting E2E entangled state to identify its quality.

First, the matter qubit stored in the quantum memory of the QBS, is stored for a duration equal to the travel time of the transferred photon plus the the required processing time. Thus, the storage time is $t_i(\textit{\textbf{l}}_r) = \frac{d_{\text{E2E},i}(\textit{\textbf{l}}_r)}{c} + T_{\text{proc}}$, where $c$ is the speed of light, $T_{\text{proc}}$ is the time needed to process a qubit after successfully transferring it to its user. We observe that the RIS location, $\textit{\textbf{l}}_r$, fundamentally affects the E2E distance, $d_{\text{E2E},i}(\textit{\textbf{l}}_r)$, that the flying qubit travels over the FSO channel, and, thus, affects the storage time of the matter qubit. During this storage time, the matter qubit suffers from \emph{quantum memory decoherence}, which we model as a \emph{depolarizing quantum noise channel}, that is the most general, worst-case quantum noise channel. The time-dependent depolarizing quantum noise channel model for a single-qubit state, represented by the density matrix $\rho$, for user $i\in\mathcal{N}$ can be written as \cite{dahlberg2019link}:
\begin{equation}\label{eq_depolarizing}
\boldsymbol{\Lambda}_{1,i}((\textit{\textbf{l}}_r)) = (1-4p_{1,i}(\textit{\textbf{l}}_r))\boldsymbol{\rho} + 4p_{1,i}(\textit{\textbf{l}}_r)\frac{\boldsymbol{I}}{2},
\end{equation}
where $p_{1,i}(\textit{\textbf{l}}_r) = \frac{1}{4} (1-e^{-\frac{t_i(\textit{\textbf{l}}_r)}{T}})$ is the depolarizing probability for the stored qubit for user $i\in\mathcal{N}$ in the QBS memory, $\boldsymbol{I}$ is the identity operator, $t_i$ is the qubit storage time, and $T$ is the coherence time of the quantum memory.\footnote{The quantum memory coherence time is a hardware characteristic of the quantum memory storing qubits. In our model, we assume that all users have the same hardware with similar coherence times.} It is clear from \eqref{eq_depolarizing} that the noise channel experienced by the stored matter qubit depends on the storage time, which is mainly affected by the RIS location, $\textit{\textbf{l}}_r$, that must be optimized.

The flying qubit experiences various types of noise when travelling through FSO quantum channels characterized with atmospheric losses, turbulence, and pointing errors, as discussed in Section \ref{section_system_model}. Atmospheric loss and pointing errors mainly cause the loss of transmitted photons, thus mainly affecting EGR. Turbulence results in both significant losses and noise. Hence, turbulence is the major source of noise on flying qubits, and results in significant phase errors, particularly when the quantum information is embedded in photon polarization \cite{vallone2014free}. Due to the nature of such turbulence-induced phase fluctuations, the resulting noise channel experienced by the flying qubit over an FSO quantum channel characterized with turbulence and other environmental losses is more suitably modeled as a \emph{phase damping quantum noise channel}, not as a general depolarizing noise channel. Accordingly, we next propose a model for the turbulence-induced phase noise error in flying qubits over FSO quantum channels, which has a direct effect on the quality of the E2E entangled states.
\begin{remark}
    The noise experienced by an entangled single photon transferred to a user $i\in\mathcal{N}$ over an FSO quantum channel characterized by losses and turbulence environmental effects, with Rytov variance $\sigma_{R,i}^2(\textit{\textbf{l}}_r)$, can be modeled by the following phase-damping quantum noise channel for a single-qubit state represented by the density matrix, $\boldsymbol{\rho}$:
    \begin{equation}\label{eq_phase_damping_noise_channel}
    \boldsymbol{\Lambda}_{2,i}(\textit{\textbf{l}}_r) = (1-p_{2,i}(\textit{\textbf{l}}_r))\boldsymbol{\rho} + p_{2,i}(\textit{\textbf{l}}_r)\boldsymbol{\sigma}_Z\boldsymbol{\rho}\boldsymbol{\sigma}_Z,
    \end{equation}
    where $\boldsymbol{\sigma}_Z$ is the Pauli-Z operator, and the probability of the phase damping quantum noise channel $p_{2,i}(\textit{\textbf{l}}_r)$ is modeled as:
    \begin{equation}\small\label{eq_phase_damping_channel_probability_parameter}
        p_{2,i}(\textit{\textbf{l}}_r) = \text{erf}(\sigma_{R,i}^2(\textit{\textbf{l}}_r)) = \frac{2}{\sqrt{\pi}} \int_{0}^{\sigma^2_{R,i}(\textit{\textbf{l}}_r)} e^{-t^2} dt,\quad\forall i \in \mathcal{N},
    \end{equation}
    where $\sigma_{R,i}^2(\textit{\textbf{l}}_r) = 1.23C_n^2k^{7/6}d_{r,i}^{11/6}(\textit{\textbf{l}}_r)$ is the Rytov variance, and $erf(.)$ is the error function of the Gaussian distribution.
\end{remark}
The noise model in \eqref{eq_phase_damping_noise_channel} and \eqref{eq_phase_damping_channel_probability_parameter} results from the fact that the Rytov variance captures phase errors due to turbulence, and it represents, without loss of generality, the variance of the resulting phase noise. This simplification allows us to focus on turbulence-induced phase fluctuations and their impact on the quality of E2E entangled links. Additionally, the turbulence-induced phase errors typically follow a Gaussian distribution in the FSOQC literature \cite{ghalaii2022quantum}. In this regard, the proposed model for the phase damping probability $p_{2,i}(\textit{\textbf{l}}_r),\forall i\in\mathcal{N}$ as a function of the Rytov variance $\sigma_{R,i}^2(\textit{\textbf{l}}_r)$ aims at connecting turbulence parameters with quantum noise channels. Particularly, this model hinges on the principle that stronger turbulence leads to a larger Rytov variance, thereby yielding a higher value of the error function. This, in turn, signifies an increased likelihood of phase damping effects on the quantum state being transmitted. Consequently, the model establishes a clear, quantifiable relationship between the intensity of atmospheric turbulence and the probability of encountering phase damping quantum noise. Although a large Rytov variance correlates with increased likelihood of significant phase disruptions, the relationship between $p_{2,i}(\textit{\textbf{l}}_r)$ and $\sigma^2_{R,i}(\textit{\textbf{l}}_r)$ is not linear but is tempered by the statistical distribution of phase errors, particularly due to the Gaussian nature of phase error distribution. Our proposed model for noise encountered during the establishment of E2E entangled connections is aligned with our loss model, and both are directly related to the environmental effects present in FSO-based QNs.

In the considered FSO-based QN, the QBS is assumed to generate pairs of entangled qubits, each of which is represented by a Bell-diagonal state, $\boldsymbol{\rho}_{BD,i} = \sum_{{j,k}\in{0,1}}\lambda_{jk,i}\boldsymbol{\Phi}_{jk}$, where $i$ corresponds to user $i\in\mathcal{N}$. The aforementioned quantum noise models capture the noise experienced while distributing an E2E entangled pair of qubits from the QBS towards its intended user. Accordingly, the resulting E2E entangled state, $\boldsymbol{\rho}''_i(\textit{\textbf{l}}_r)$, between the QBS and user $i$ is derived in the following proposition.

\vspace{-0.1cm}
\begin{proposition}\label{e2e_state_proposition}
The density matrix characterizing a successfully distributed E2E entangled state in the proposed framework is given by:
\begin{equation}\label{eq_final_state_lemma}
\begin{split}
    \boldsymbol{\rho}_{i}''(\textit{\textbf{l}}_r) &= \lambda_{00,i}''(\textit{\textbf{l}}_r)\boldsymbol{\Phi}_{00} + \lambda_{01,i}''(\textit{\textbf{l}}_r)\boldsymbol{\Phi}_{01}\\ &+ \lambda_{10,i}''(\textit{\textbf{l}}_r)\boldsymbol{\Phi}_{10} + \lambda_{11,i}''(\textit{\textbf{l}}_r)\boldsymbol{\Phi}_{11},
\end{split}
\end{equation}
with
\begin{equation}\footnotesize\label{equation_E2E_fidelity}
    \begin{split}
         \lambda_{00,i}''(\textit{\textbf{l}}_r) &= \Big[\left(1-p_{2,i}(\textit{\textbf{l}}_r)\right)F_{00}(\textit{\textbf{l}}_r) + p_{2,i}(\textit{\textbf{l}}_r)F_{01}(\textit{\textbf{l}}_r)\Big]\\
         \lambda_{01,i}''(\textit{\textbf{l}}_r) &= \Big[\left(1-p_{2,i}(\textit{\textbf{l}}_r)\right)F_{01}(\textit{\textbf{l}}_r) + p_{2,i}(\textit{\textbf{l}}_r)F_{00}(\textit{\textbf{l}}_r)\Big]\\
         \lambda_{10,i}''(\textit{\textbf{l}}_r) &= \Big[\left(1-p_{2,i}(\textit{\textbf{l}}_r)\right)F_{10}(\textit{\textbf{l}}_r) + p_{2,i}(\textit{\textbf{l}}_r)F_{11}(\textit{\textbf{l}}_r)\Big]\\
         \lambda_{11,i}''(\textit{\textbf{l}}_r) &= \Big[\left(1-p_{2,i}(\textit{\textbf{l}}_r)\right)F_{11}(\textit{\textbf{l}}_r) + p_{2,i}(\textit{\textbf{l}}_r)F_{10}(\textit{\textbf{l}}_r)\Big],
    \end{split}
\end{equation}
where $p_{2,i}(\textit{\textbf{l}}_r)= \text{erf}(\sigma_R^2(\textit{\textbf{l}}_r))$, $F_{jk}(\textit{\textbf{l}}_r) = \left(\frac{1}{4} + \left(\lambda_{jk,i}-\frac{1}{4}\right)e^{-\frac{t_i(\mathbf{l}_r)}{T}} \right)$, and $F_{00}(\textit{\textbf{l}}_r) + F_{01}(\textit{\textbf{l}}_r) + F_{10}(\textit{\textbf{l}}_r) + F_{11}(\textit{\textbf{l}}_r)= 1$. 
\end{proposition}

\begin{proof}
See Appendix \ref{appendix_lemma_final_state_derivation}.
\end{proof}

Proposition \ref{e2e_state_proposition} allows us to analyze the quality of the established E2E entanglement between the QBS and end nodes in the QN. In particular, the fidelity of the E2E state is characterized by the probability of obtaining the $\boldsymbol{\Phi}_{00}$ state when measuring the E2E state, and thus is given by the parameter: $\bra{\Phi_{00}}\boldsymbol{\rho}_{\mathrm{BD},i}\ket{\Phi_{00}} = \lambda_{00,i}''(\textit{\textbf{l}}_r)$. In this regard, we note the importance of the RIS placement location, $\textit{\textbf{l}}_r$, on the quality of E2E entangled states, which must satisfy a minimum QoS requirement in order for such states to be useful in quantum applications.



\vspace{-0.3cm}
\section{Joint RIS Placement and EGR Allocation Optimization Problem}\label{sec_optimization}
The analysis presented in Section \ref{sec_III}, comprehensively characterizes the distributed E2E entangled states, their quality, and generation rates to each user in the QN. Furthermore, Section \ref{sec_III} clearly demonstrated the important role of RIS placement on overall QN performance. Accordingly, in this section, we formulate a novel optimization problem that integrates optimal RIS placement with EGR resource allocation in an FSO-based QN in the presence of environmental effects. 

To do so, the initial EGR $R_{\mathrm{in},i}$ assigned by the QBS to each user $i\in\mathcal{N}$ must be optimized such that the overall EGR does not exceed the quantum memory capacity limits of the QBS. Moreover, the optimal RIS placement is characterized by identifying the optimal 3D RIS location $\mathbf{l}_r = (x_r, y_r, H_r)$ within the permitted deployment region $\Omega$, described in \eqref{eq_feasible_RIS_region}. Accordingly, the optimization variables are 1) $\boldsymbol{l}_r$: the location vector of the RIS, and 2) $\boldsymbol{R}_{\text{in}}$: the vector incorporating the initial EGRs $R_{\mathrm{in},i}$ for all users $i\in\mathcal{N}$.

The achieved E2E EGR for each user must satisfy a minimum E2E rate constraint. For guaranteeing fairness among the different users, the entangled qubits remaining available, after satisfying the users' minimum rate requirements, must be fairly distributed among all QN user. To quantify this fairness, we use the \emph{Jain's fairness index (JFI)} \cite{eger2007fair}:
\begin{equation}
    U_{\mathrm{JFI}}(\boldsymbol{l}_r,\boldsymbol{R}_{\text{in}}) = \frac{\left( \sum_{i\in\mathcal{N}} R_{\mathrm{E2E},i}(\textit{\textbf{l}}_r,R_{\mathrm{in},i}) \right)^2}{N \cdot \sum_{i\in\mathcal{N}} R^2_{\mathrm{E2E},i}(\textit{\textbf{l}}_r,R_{\mathrm{in},i})},
\end{equation}
where $R_{\mathrm{E2E},i}(\textit{\textbf{l}}_r,R_{\mathrm{in},i})$ corresponds to the E2E EGR associated with user $i\in\mathcal{N}$. However, JFI does not account for \emph{weighted fairness}, where some users in the QN have certain priority over others, e.g., their corresponding quantum applications require a higher E2E EGR compared to other QN users. Thus, such users are given higher corresponding weights. To capture these cases, we adopt the extended version of the JFI, \emph{weighted fairness index (WFI)} \cite{eger2007fair}, as our utility function. This index takes a maximum value of 1 when all users are fairly served with E2E entangled qubits based on their priorities. The corresponding objective function based on WFI is: 
\begin{equation}\label{eq_objective_func_U}
U_{\mathrm{WFI}}\left(\boldsymbol{l}_r,\boldsymbol{R}_{\text{in}}\right) = \frac{\left( \sum_{i\in\mathcal{N}} R_{\mathrm{E2E},i}\left(\textit{\textbf{l}}_r,R_{\mathrm{in},i}\right) \right)^2}{\sum_{i\in\mathcal{N}} \left( \frac{R^2_{\mathrm{E2E},i}(\boldsymbol{l}_r,R_{\mathrm{in},i})}{w_i} \right)},
\end{equation}
where $w_i$ represents the weight associated with user $i\in\mathcal{N}$, a parameter accounting for the relative importance of the quantum application performed by user $i$ in the QN. WFI is able to achieve very high fairness levels between QN users, while accounting for the importance of the user quantum application.  

Finally, different QN users may have different QoS requirements on the minimum required fidelity of E2E entangled states, along with minimum requirements on E2E EGRs. As such, the proposed optimization problem can be formulated as follows:

\vspace{-0.4cm}
\begin{subequations}
\begin{alignat}{2}
\mathcal{P}1: \quad &\!\max_{\boldsymbol{R_{\text{in}}},\boldsymbol{l}_r}        &\quad& \sum_{i\in\mathcal{N}}w_iR_{\text{E2E},i}  \label{eq:optProb}\\
&   s.t.               &      & \sum_{i\in\mathcal{N}}R_{\text{in},i} \leq C_{\text{max}} \quad,\forall i\in \mathcal{N},\label{eq:constraint1}\\
&                  &      & R_{\text{E2E},i} \geq R_{\text{min},i} \quad,\forall i\in \mathcal{N},\label{eq:constraint2}\\
&                  &      & U_{\mathrm{WFI}}(\boldsymbol{l}_r,\boldsymbol{R}_{\text{in}}) \geq \delta_{\mathrm{th}} \quad,\forall i\in \mathcal{N},\label{eq:constraint3}\\
&                  &      & \lambda_{00,i}''(\textit{\textbf{l}}_r) \geq f_{\text{min},i} \quad,\forall i\in \mathcal{N},\label{eq:constraint4}\\
&                  &      & x_{\text{min}}\leq x_r \leq x_{\text{max}},\label{eq:constraint5}\\
&                  &      & y_{\text{min}}\leq y_r \leq y_{\text{max}},\label{eq:constraint6}\\
&                  &      & H_{\text{min}}\leq H_r \leq H_{\text{max}},\label{eq:constraint7}
\end{alignat}
\end{subequations}
where constraint \eqref{eq:constraint1} ensures that the overall initial EGR does not exceed the maximum rate capacity of the QBS's quantum memory. Constraint \eqref{eq:constraint2} ensures that the E2E EGR of user $i\in\mathcal{N}$ lies above a minimum required value $R_{\text{min},i}$. Moreover, constraint \eqref{eq:constraint3} guarantees fairness in the EGR distribution among the different users by making sure the WFI is above a minimum required threshold, $\delta_{\mathrm{th}}$. Constraint \eqref{eq:constraint4} guarantees that the fidelity, $\lambda_{00,i}''(\textit{\textbf{l}}_r)$, of E2E entangled links associated with user $i\in\mathcal{N}$ lies above their minimum fidelity threshold, $f_{\text{min},i}$. Finally, constraints \eqref{eq:constraint3},\eqref{eq:constraint4}, and \eqref{eq:constraint5} ensure that the RIS placement location falls within the corresponding physical boundaries along the x-,y-, and z-axes, respectively. 

The proposed optimization problem $\mathcal{P}1$ is a non-convex programming problem, that is, in general, \emph{NP}-hard to solve. In order to solve $\mathcal{P}1$, and since the derivatives of the functions in (\ref{eq:optProb}) and (\ref{eq:constraint2}) are not easily computed, it is typical to consider a derivative-free metaheuristic solution. Thus, we develop a metaheuristic solution based on simulated annealing, as described in Algorithm \ref{algorithm_simulated_annealing}.


\vspace{0.3cm}
\begin{algorithm}[t!]\small
\caption{Simulated Annealing Algorithm for $\mathcal{P}1$}\label{algorithm_simulated_annealing}
\begin{algorithmic}[1]
\State Initialize the current solution $\boldsymbol{R}_{\text{in}}$, and $\textit{\textbf{l}}_r$ to a random feasible solution in the search space
\State Define $L$ as the number of iterations per temperature level
\State Set initial temperature $T_{\text{SA}}=T_0$
\State Initialize the optimal solution $\boldsymbol{R}_{\text{in}}^*$, and $\textit{\textbf{l}}_r^*$ to the current solution
\While{$T_{\text{SA}}>T_{\min}$}
\For{$i=1$ to $L$}
\State Generate random neighbors $\boldsymbol{R}_{\text{in}}'$, and $\textit{\textbf{l}}_r'$ of $\boldsymbol{R}_{\text{in}}$, and $\textit{\textbf{l}}_r$, respectively
\If{$x_r' \in [x_{\text{min}}, x_{\text{max}}]$, $y_r' \in [y_{\text{min}}, y_{\text{max}}]$, $H_r' \in [H_{\text{min}}, H_{\text{max}}]$, $\sum_{i\in\mathcal{N}}R_{\text{in},i} \leq C_{\text{max}}$, $R_{\text{E2E},i} \geq R_{\text{min},i}$, $U_{\mathrm{WFI}}(\boldsymbol{l}_r,\boldsymbol{R}_{\text{in}}) \geq \delta_{\mathrm{th}}$, and $\lambda_{00,i}''(\textit{\textbf{l}}_r) \geq f_{\text{min},i} \quad,\forall i\in \mathcal{N}$}
\State Calculate $\Delta U = 
U_{\mathrm{WFI}}(\textit{\textbf{l}}_r',\boldsymbol{R}_{\text{in}}') -U_{\mathrm{WFI}}(\textit{\textbf{l}}_r,\boldsymbol{R}_{\text{in}})$
\If{$\Delta U>0$}
\State Accept $\boldsymbol{R}_{\text{in}}'$, and $\textit{\textbf{l}}_r'$ as the new solution
\Else
\State Generate a random number $r\in [0,1]$
\If{$e^{\Delta U/T_{\text{SA}}}>r$}
\State Accept $\boldsymbol{R}_{\text{in}}'$, and $\textit{\textbf{l}}_r'$ as the new solution
\Else
\State Reject $\boldsymbol{R}_{\text{in}}'$, and $\textit{\textbf{l}}_r'$, and keep $\boldsymbol{R}_{\text{in}}$, and $\textit{\textbf{l}}_r$ as the current solution
\EndIf
\EndIf
\State Update the optimal solution $\boldsymbol{R}_{\text{in}}^*$, and $\textit{\textbf{l}}_r^*$ if $U_{\mathrm{WFI}}(\textit{\textbf{l}}_r,\boldsymbol{R}_{\text{in}}) > U_{\mathrm{WFI}}(\textit{\textbf{l}}_r^*,\boldsymbol{R}_{\text{in}}^*)$
\EndIf
\EndFor
\State Update temperature $T_{\text{SA}} = T_{\text{SA}} \cdot \alpha_{\text{SA}}$ (exponential cooling)
\EndWhile
\State \textbf{return} the best solution found, $\boldsymbol{R}_{\text{in}}^*$, and $\textit{\textbf{l}}_r^*$
\end{algorithmic}
\end{algorithm}\vspace{-0.15cm}

In Algorithm \ref{algorithm_simulated_annealing}, we initialize the optimization variables $\boldsymbol{R}_{\text{in}}$, and $\textit{\textbf{l}}_r$ to random feasible solutions in the search space. We set the initial temperature $T_{\text{SA}}=T_0$, and initialize the best solutions $\boldsymbol{R}_{\text{in}}^*$, and $\textit{\textbf{l}}_r^*$ to the current solutions.

In each iteration of Algorithm \ref{algorithm_simulated_annealing}, we generate random neighbors $\boldsymbol{R}_{\text{in}}'$, and $\textit{\textbf{l}}_r'$ of the current solutions calculate the change in the objective function, $\Delta U$. Solutions with $\Delta U>0$ are accepted, enhancing the search. For $\Delta U \leq 0$, acceptance is probabilistic, adhering to the Metropolis criterion $\exp(\Delta U/T_{\text{SA}})$, and compared with a random number $r \in [0,1]$. This method helps to avoid local minima and ensures rigorous exploration of the search space.

Moreover, if the new accepted solution achieves a larger WFI value compared to the current best solution, the algorithm updates $\boldsymbol{R}_{\text{in}}^*$ and $\textit{\textbf{l}}_r^*$ with these values. In each iteration, the temperature undergoes a reduction based on an exponential decay model until it reaches the minimum threshold, $T_{\min}$. Upon completion, the algorithm outputs the best-found solutions, $\boldsymbol{R}_{\text{in}}^*$ and $\textit{\textbf{l}}_r^*$, which are approximate solutions that are generally near optimal \cite{kirkpatrick1983optimization}.

\section{Simulation Results and Analysis}
Next, we perform extensive simulations to evaluate the performance of the proposed RIS-assisted FSO-based QN. For our simulations, we define the following \emph{default QN setup}: A QBS is located at $\textit{\textbf{l}}_s = (0,0,90)$\,m in a 3D grid, serving a number of QN users whose locations vary along the x- and y-axes, with a fixed height of $10$\,m. User locations are non-uniformly distributed, following a truncated normal distribution, within a $400$\,m$\times400$\,m squared region. Each user $i$ is located at $\textit{\textbf{l}}_i = (x_i, y_i, 10)$\,m, where $x_i \sim \mathcal{TN}(250, 50, 50, 450)$\,m with a mean of 250 and a standard deviation of 50, and $y_i \sim \mathcal{TN}(200, 50, 0, 400)$\,m, with a mean of 200 and a standard deviation of 50, $\forall i\in\mathcal{N}$. Statistical results are averaged over 1000 simulation runs, and the region, $\Omega$, within which the RIS can be deployed corresponds to the same squared user area, defined as $\Omega = \{(x_r, y_r, H_r) \mid 50$\,m $\leq x_r \leq 450$\,m, $0$\,m $\leq y_r \leq 400$\,m, $35$\,m $\leq H_r \leq 90$\,m\}. The RIS placement location $\boldsymbol{l}_r$ is optimized for each set of user locations, with a minimum distance of $20$\,m required between the RIS and each user to avoid being co-located. The feasible range of values for the initial EGR $R_{\text{in},i}$ for a user $i\in\mathcal{N}$ is: 1\,kHz-1\,MHz~\cite{chehimi2021entanglement_rate_optimization,panigrahy2023capacity}. Here, we adopt the $\alpha_s$-model for entanglement generation \cite{humphreys2018deterministic,rozpedek2019building}. Then, a twirling map is assumed to be applied to transform the generated entangled pairs at the QBS into identical Bell-diagonal quantum states. In particular, the initial entanglement generation attempt rate is considered to be $1$\,MHz \cite{panigrahy2023capacity}, with a probability of success for the initial entangled pair generation process $p_{\mathrm{in}} = 2\alpha_s$, where $\alpha_s$ is a tunable parameter that enables the rate-fidelity tradeoff \cite{humphreys2018deterministic}. We thus set the initial EGR, $R_{\text{in},i} = 2\alpha_s\times 10^6$. Note that parameter $\alpha_s\in(0,0.5)$ is directly related to the initial fidelity of the generated entangled pairs, $\lambda_{00,i}$, since $\lambda_{00,i} = 1 - \alpha_s$. As such, for an initial fidelity in the practical range $\lambda_{00,i} \in(0.5,0.9995)$, the range of $\alpha_s$ is adjusted to $\alpha_s\in(0.0005,0.5)$. In particular, a very high initial EGR $R_{\text{in},i} = 1$\,MHz corresponds to entangled qubits with a low initial fidelity $\lambda_{00,i} = 0.5$, while a low initial EGR $R_{\text{in},i} = 1$\,kHz results in entangled qubits with a high initial fidelity $\lambda_{00,i} = 0.9995$. This is due to the rate-fidelity tradeoff within the adopted entanglement generation scheme, whereby the fidelity of generated entangled qubits is compromised when a higher EGR is desired \cite{humphreys2018deterministic}. Unless specified otherwise, the users are assumed to be equally-weighted, have similar minimum rate requirement of $R_{\mathrm{min}} = 1$ \cite{chehimi2023scaling} and a minimum required fairness $\delta_{\text{th}} = 0.95$, and their required minimum E2E fidelity is sampled from a uniform distribution such that $\lambda_{00,i} \sim \mathcal{U}(0.5,0.7)$ \cite{panigrahy2023capacity,chehimi2023scaling}, which correspond to different quantum applications. The different environmental parameters considered in the default setup, which correspond to sunny weather, moderate turbulence, and low pointing error are summarized in Table \ref{tbl:Def}, based on \cite{chiti2022mobile,9443170,alshaer2021hybrid,10130384,9466323,herbschleb2019ultra,chehimi2023matching}, and \cite{dahlberg2019link}. Unless stated otherwise, the default setup parameters are adopted throughout the simulations.
\begin{table}[t!] 
\caption{Summary of Parameters used in Simulations.}\vspace{-0.2cm}
\label{tbl:Def}\small
\centering
\begin{tabular}{|p{1.3cm}|p{3.5cm}|p{2.5cm}|}
 \hline
 {\bf Parameter} & {\bf Description} & {\bf Value} \\ \hline
 $\lambda$ & Wavelength. & 1550\,nm\footnotemark \\ \hline 
 $\zeta_{th}$ & FSO channel threshold. & 0.05\\ \hline 
 $\eta$ & Responsivity of the receiver. & 0.95 \\ \hline 
 $\varsigma$ & RIS reflection efficiency. & 0.97 \\ \hline 
 $\hat{\kappa}$ & Weather-dependent attenuation coefficient of the FSO link. & $0.43$ dB/km  \\ \hline 
 $C_n^2$ & Refractive index structure constant which quantifies turbulence strength. & moderate turbulence: $5\times 10^{-14}$ m$^{-2/3}$, strong turbulence: $1\times 10^{-13}$ m$^{-2/3}$ \\ \hline 
 $r_a$ & Aperture radius. & 0.55\,m   \\ \hline 
 $\phi_d$ & Beam divergence angle. & $8$\,mrad \\ \hline 
 $\sigma_{\theta}$ & Standard deviation for misalignment angle $\theta$. & $1$\,mrad  \\ \hline 
 $\sigma_{\phi}$ & Standard deviation for misalignment angle $\phi$. & $0.25$\,mrad \\ \hline 
$C_\text{max}$ & Quantum memory capacity & $1\times 10^{7}$\\ \hline  
$T$ & Quantum memory coherence time & $2.43$\,ms\\ \hline 
$T_{\text{proc}}$ & Average qubit processing time & $4\,\mu$s\\ \hline 
\end{tabular} 
\end{table}
\footnotetext{Note that atmospheric absorption becomes almost negligible at this wavelength, particularly in clear air conditions, thus, it was selected for FSOQC.}


First, we study the impact of varying the locations of QN users on the optimal RIS location. In particular, we consider a QN with three users placed in two different scenarios: 1) the three users are close to each other, and separated by equal distances of $50$\,m, i.e., $\boldsymbol{l}_1 = (350,0,10)$, $\textit{\textbf{l}}_2 = (400,0,10)$, and $\boldsymbol{l}_3 = (450,0,10)$, 2) the first user is moved closer to the QBS, such that $\boldsymbol{l}_1 = (200,0,10)$, thus one user is far away from the other two. For each scenario, we consider two cases, where the weights of the different users are varied. The first case assigns equal weights to all users, $\boldsymbol{w}=[1/3, 1/3, 1/3]$, while the second case assigns weights that increase with the distance of the users from the QBS, $\boldsymbol{w} = [0.1, 0.3, 0.6]$. We observe from Fig.~\ref{fig_result_fig1_RIS_placement} that when the three users are close to one another, the optimal RIS placement location will be to the left of the first user, particularly, at an intermediate point along the x-axis between the QBS and its closest user, while satisfying the minimum y-axis 20\,m distance. Here, when the weights of the users are changed such that the farthest away user has a higher priority, the RIS shifts to the right closer towards the users. On the other hand, when one of the users is closer to the QBS, the RIS is placed in-between the users. In this case, as the weights change, the RIS shifts towards the user with the smallest weight, because moving towards the other users with high weights would yield a very low E2E rate for the first user, which would correspond to a lower sum rate due to the fairness conditions. For this scenario, we observe that our proposed simulated annealing algorithm achieves a near optimal performance within 6\% of the optimal solution obtained through exhaustive search, while reducing the required computational time by more than 65\%.

\begin{figure}
    \centering
    \includegraphics[width=\linewidth]{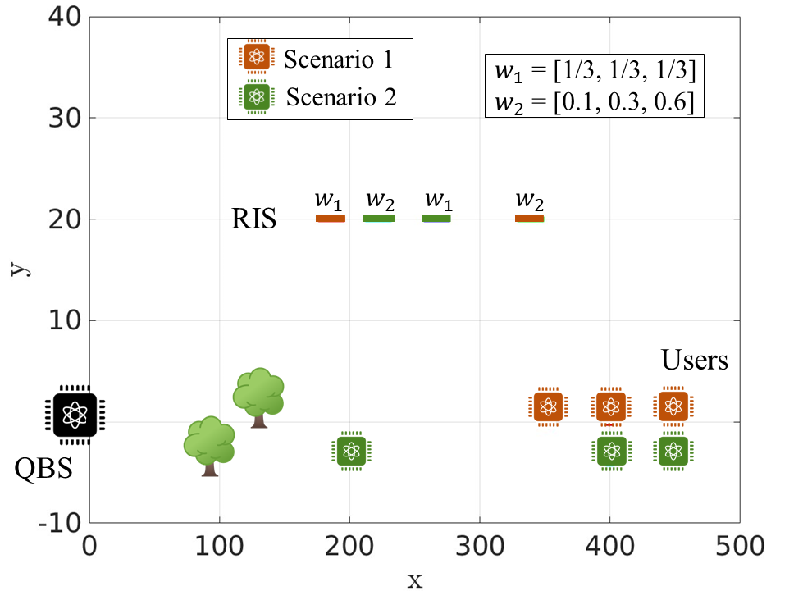}\vspace{-0.15cm}
    \caption{Optimal RIS placement for different scenarios of user distribution.}\vspace{-0.5cm}
    \label{fig_result_fig1_RIS_placement}
\end{figure}

Next, in Fig. \ref{fig_results_fig3_E2E_rate_RIS}, the achieved E2E EGRs, $R_{\text{E2E}}$, for a three-user QN in the aforementioned first scenario is considered. In this figure, we compare performance with a \emph{rate maximizing} framework where the minimum fairness threshold constraint is relaxed, and a \emph{log-rate maximizing} framework, where both the minimum fairness and rate constraints are relaxed, and the optimization objective is to maximize the weighted sum of the log-rate $\sum_{i\in\mathcal{N}}w_i\log{R_{\text{E2E},i}}$. We clearly observe from Fig. \ref{fig_results_fig3_E2E_rate_RIS} that our proposed framework increases fairness (WFI) by around 64\% compared to the rate and log-rate maximizing frameworks. Fig. \ref{fig_results_fig3_E2E_rate_RIS} also shows that the E2E sum rate achieved by the rate and log-rate maximizing frameworks is higher than the one obtained using our framework. This is because these two frameworks provide the user with the best channel condition, i.e., shortest E2E distance and weather condition, with the maximum possible E2E rate, while other more distant users are given lower rates that just satisfy their minimum rate and fidelity constraints. In particular, we observe from Fig. \ref{fig_results_fig3_E2E_rate_RIS} that user 1, which is the closest to the QBS and has the shortest E2E distance, receives a significantly higher E2E rate compared to other users in the rate and log-rate maximizing frameworks. Thus, the resulting E2E sum rate is high, while fairness is sacrificed. In contrast, our framework first ensures that all users satisfy their minimum rate constraint, and, then, allocates the remaining resources fairly among all users. Fig. \ref{fig_results_fig3_E2E_rate_RIS} also shows that the rate and log-rate maximizing frameworks have a very similar performance. Additionally, the rate maximizing framework can be considered as a special case of our framework, corresponding to the case when the minimum fairness threshold, $\delta_{\text{th}}$, is set to zero. Furthermore, the results shown in Fig. \ref{fig_results_fig3_E2E_rate_RIS} were, again, near-optimal, within less than 6\% of the optimal value obtained using exhaustive search.

\begin{figure}
    \centering
 \includegraphics[width=\linewidth]{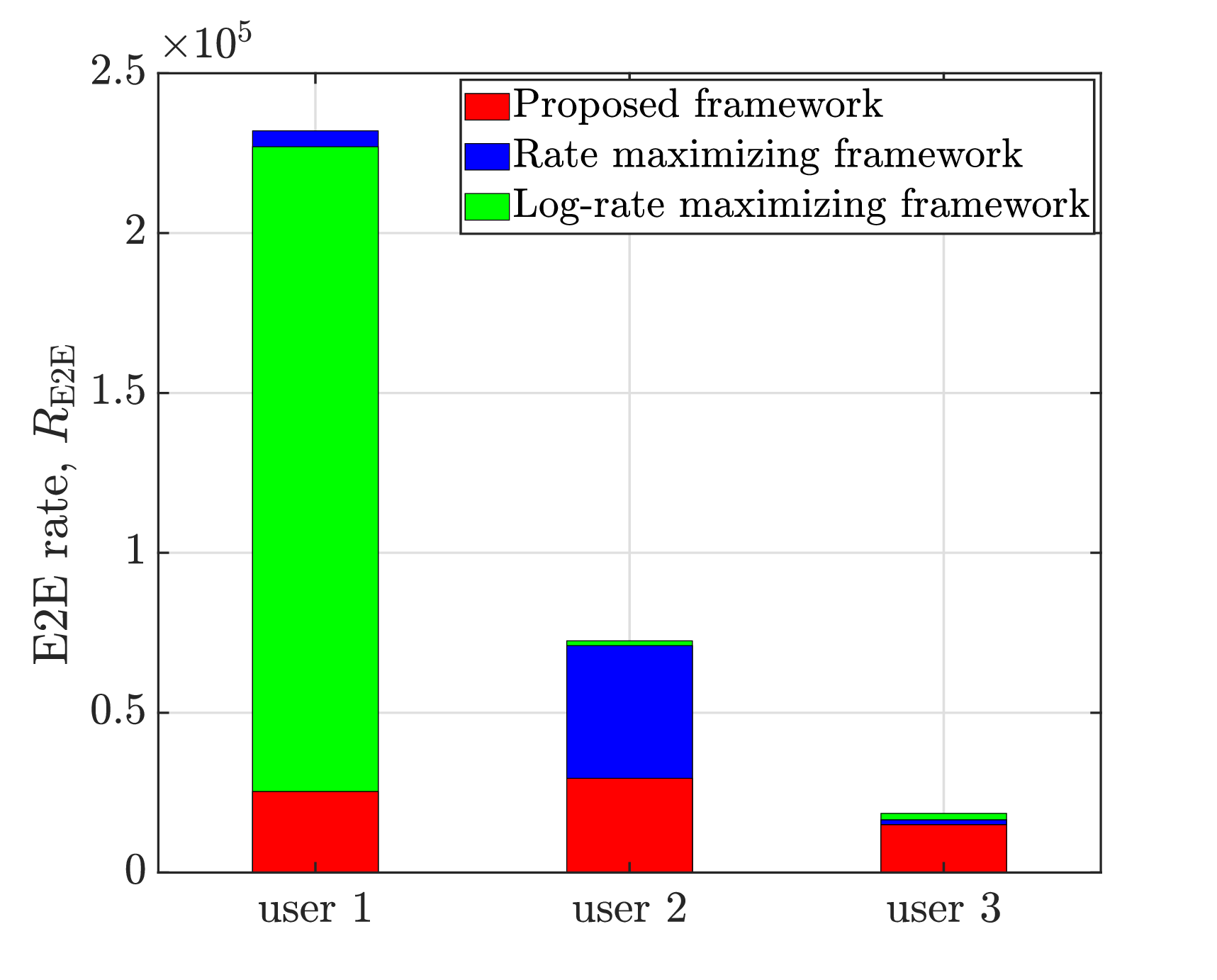}\vspace{-0.3cm}
    \caption{E2E EGR achieved at the optimal RIS location for a QN with three users located based on scenario 1 in Fig. \ref{fig_result_fig1_RIS_placement}.}\vspace{-0.6cm}
    \label{fig_results_fig3_E2E_rate_RIS}
\end{figure}

\begin{figure}[!t]\vspace{-0.3cm}
    \centering
    \includegraphics[width=0.5\textwidth]{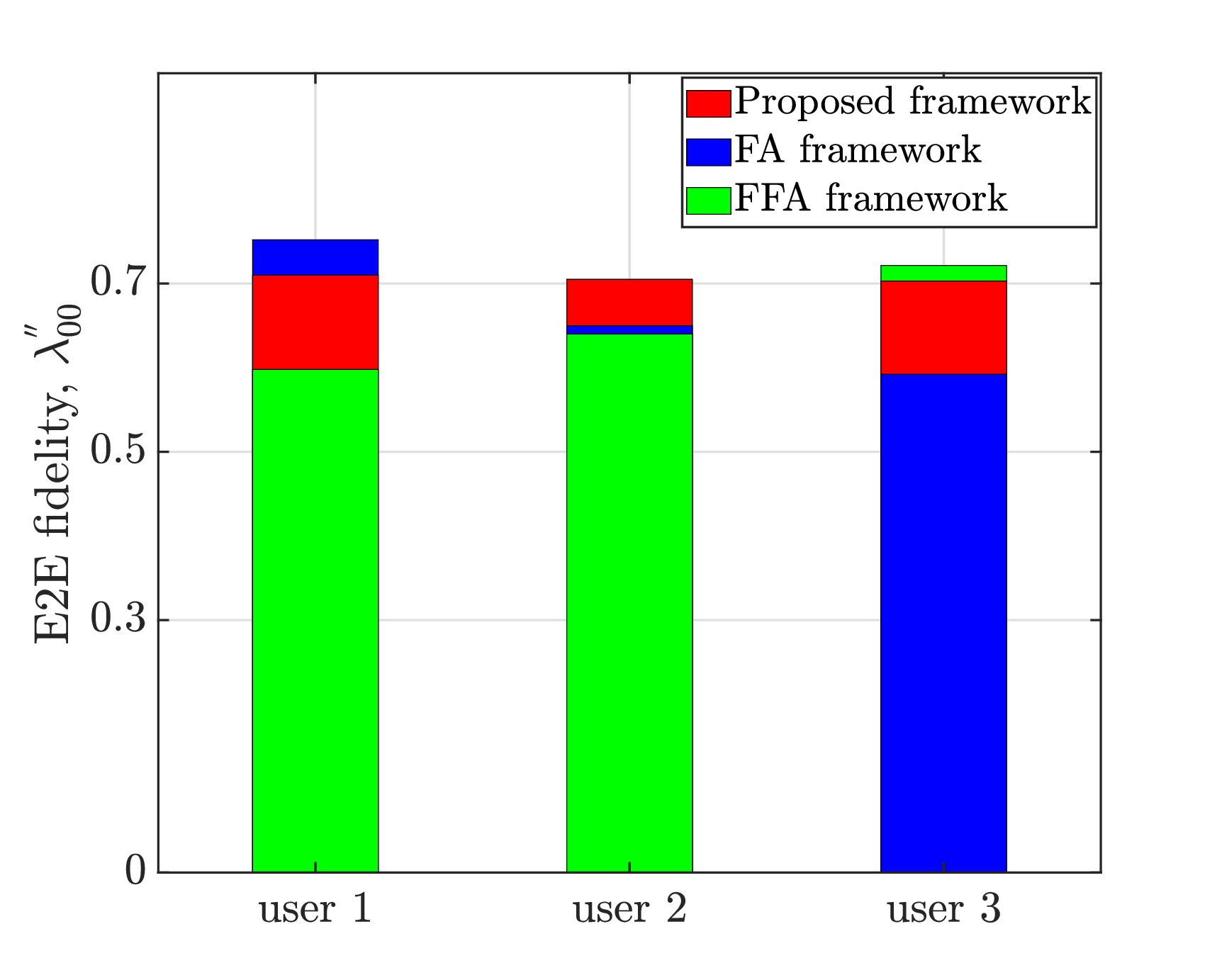}\vspace{-0.15cm}\\
    (a) Moderate turbulence.\\[1ex] \vspace{-0.15cm}
    \includegraphics[width=0.5\textwidth]{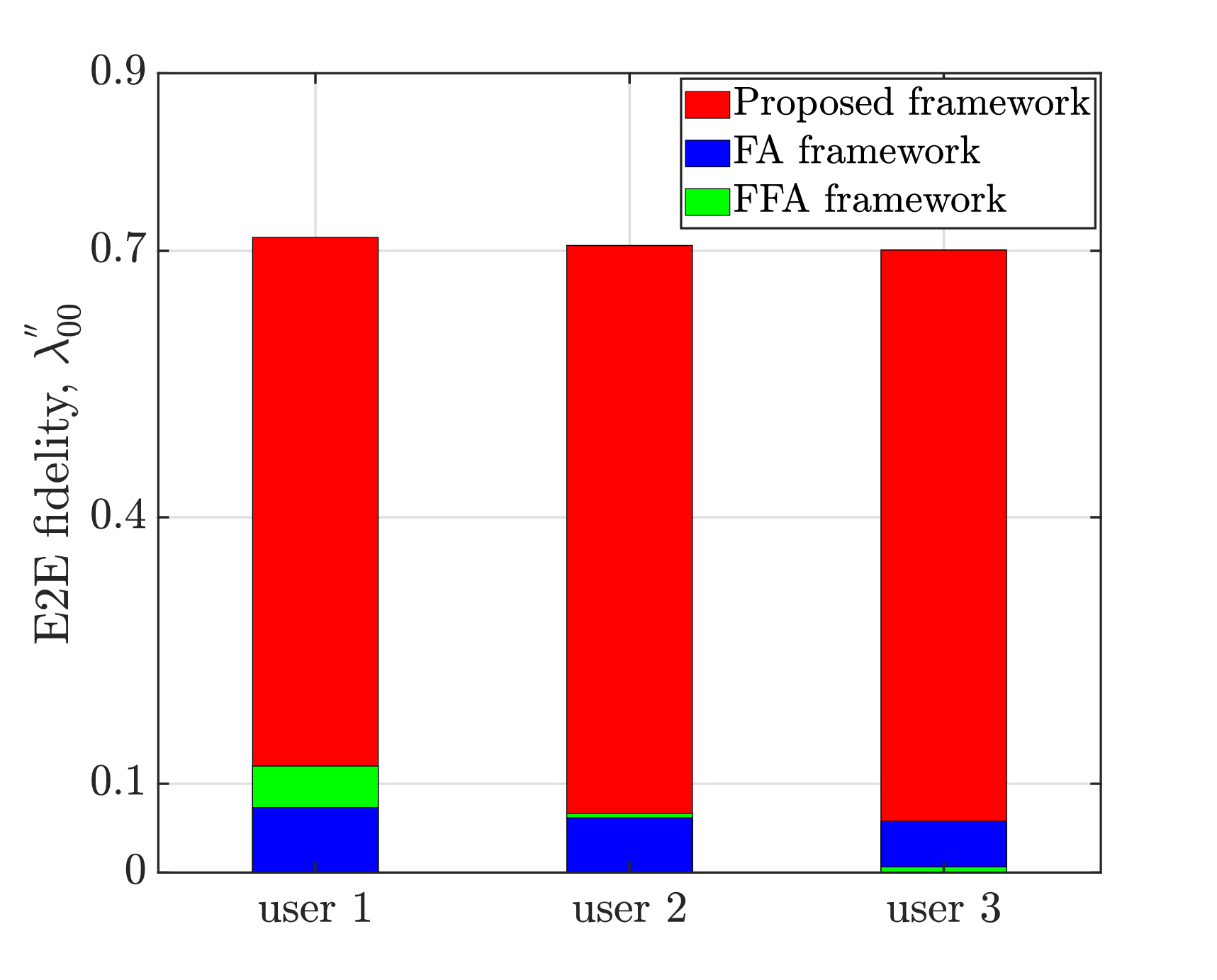}\vspace{-0.15cm}\\
    (b) Strong turbulence. \vspace{-0.15cm}
    \caption{Achieved E2E fidelity in a 3-user QN under different turbulence effects.}\vspace{-0.3cm}
    \label{fig_E2E_fidelity_three_users}
    \vspace{-0.4cm}
\end{figure}

Our framework uniquely models quantum noise effects on entangled qubits in FSO channels, ensuring distributed entangled states meet minimum fidelity thresholds for effective quantum applications. This is in a stark contrast to classical frameworks that overlook resource fidelity. To validate this, in Fig. \ref{fig_E2E_fidelity_three_users}, we show the obtained E2E fidelity, $\lambda_{00,i}''$, for the users in a three-user QN, where users perform a quantum application that requires entangled qubits with a minimum fidelity $\delta_{\text{th}} = 0.7$. We perform Monte Carlo simulations for the locations, and we consider user 3 as having the largest average distance from the QBS. To asses the performance of our proposed framework, we compare the obtained results against two benchmarks: 1) \emph{fidelity-agnostic (FA) framework}, which is the classical framework for resource allocation in communication networks, where resources are distributed while disregarding fidelity, and 2) \emph{fidelity- and fairness-agnostic (FFA) framework}, where both the minimum required fidelity and fairness constraints are relaxed in our framework. Fig. \ref{fig_E2E_fidelity_three_users} also shows the obtained E2E fidelity under both moderate and strong turbulence effects. First, under moderate turbulence effects, in Fig. \ref{fig_E2E_fidelity_three_users}(a), we observe that the proposed framework is the only one that manages to achieve the minimum required fidelity for all three users in the QN. Moreover, Fig. \ref{fig_E2E_fidelity_three_users}(a) shows that the FA framework results in the highest fidelity for user 1, which has the shortest average E2E distance, and in the lowest fidelity for user 3, which has the longest average E2E distance from the QBS. This is because the FA framework is designed to ensure fairness among users, thereby limiting the E2E rates. Consequently, this restriction affects the initial EGR assigned to user 1, who experiences the least amount of losses in the QN. Meanwhile, the highest initial EGR will be assigned to user 3 to account for its high losses. Due to the inverse relation between the initial EGR and initial entanglement fidelity, the initial fidelity of the entangled qubits assigned to user 1 will be the highest, while user 3 will be assigned entangled qubits with the lowest E2E fidelity. The exact opposite performance is obtained using the FFA framework, because this framework maximizes the E2E rate, assigned to the user with the shortest E2E distance, i.e., user 1. Accordingly, user 1 achieves the lowest E2E fidelity. Additionally, this framework ensures that users with the largest E2E distances, and correspondingly highest losses, e.g., user 3, just satisfy their minimum required E2E rate and, thus, achieve the highest E2E fidelity. 

Next, in Fig. \ref{fig_E2E_fidelity_three_users}(b), under strong turbulence conditions, which correspond to harsh environmental effects, we observe that the two considered classical frameworks, which are agnostic to fidelity in resource allocation, fail to satisfy the minimum fidelity requirements of the QN users. In contrast, our proposed framework effectively meets the heterogeneous minimum fidelity requirements of all QN users. In particular, the distributed qubits using the baseline FA framework and the FFA framework have a fidelity that is at least 83\% less than the users' required minimum fidelity under strong turbulence conditions. Therefore, in order for these frameworks to successfully distribute entangled qubits over the QN users, they must sacrifice the fidelity of those qubits, which results in useless entangled qubits that cannot perform practical quantum applications. The reason why strong turbulence, specifically, has such a significant effect on the E2E fidelity is because it results in phase fluctuations that directly lead to high quantum noise effects, an effect that our proposed framework uniquely captures, as discussed in Section \ref{sec_quantum_noise}.

\begin{figure}[!t]
    \centering
    \includegraphics[width=0.5\textwidth]{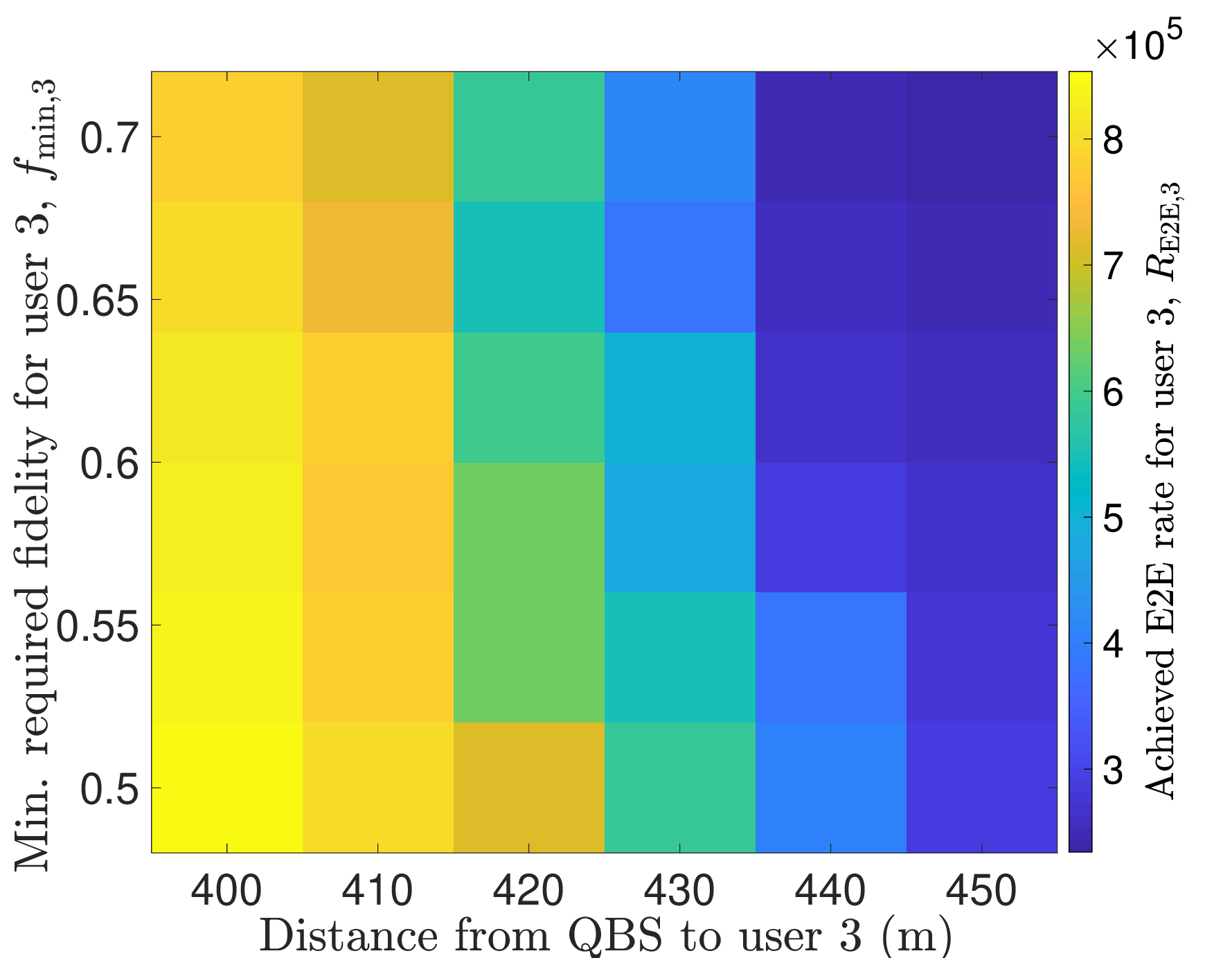}\vspace{-0.15cm}\\
    (a) Moderate turbulence.\\[1ex] \vspace{-0.15cm}
    \includegraphics[width=0.5\textwidth]{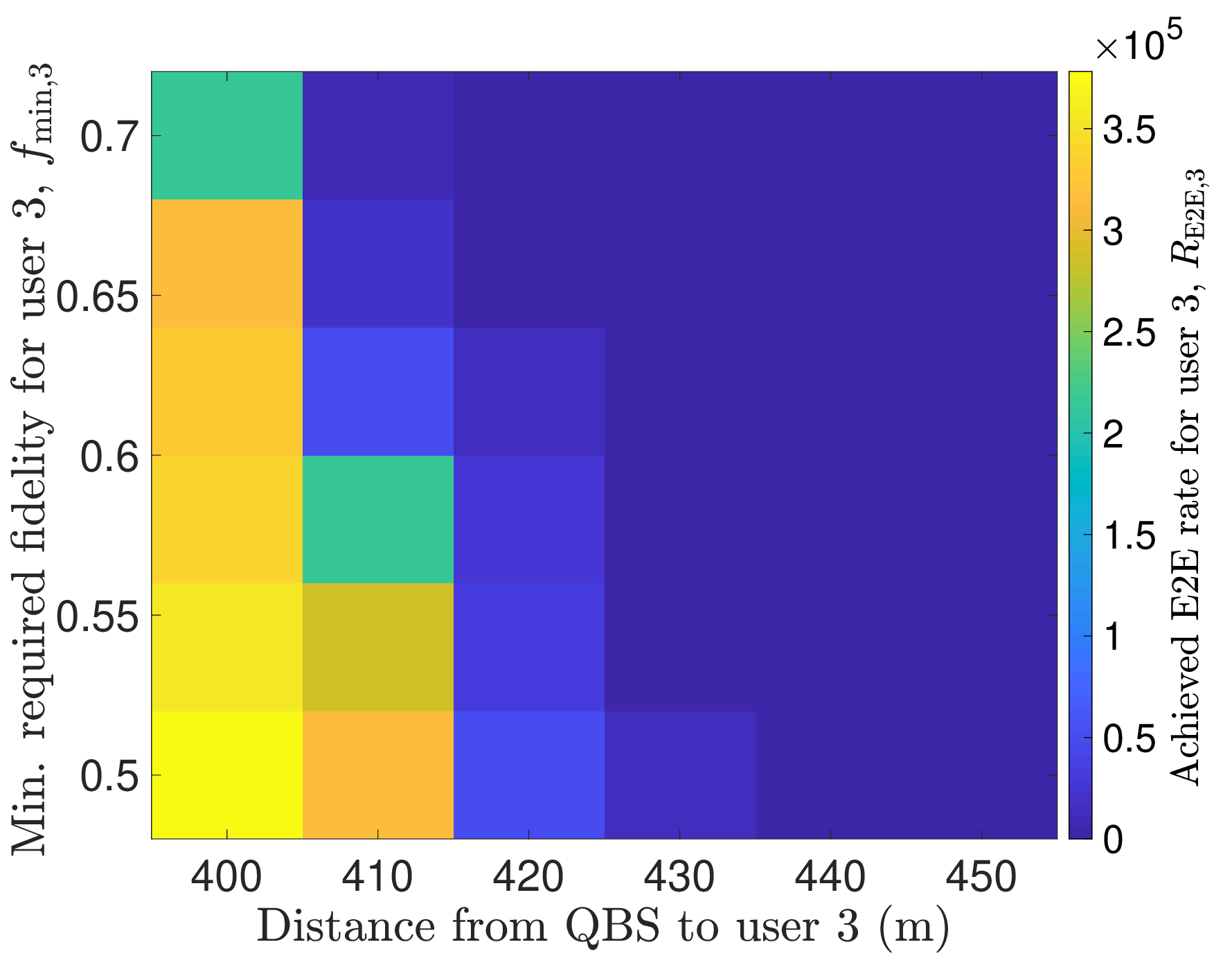}\vspace{-0.15cm}\\
    (b) Strong turbulence. \vspace{-0.15cm}
    \caption{Achieved E2E rate for the third user in a 3-user QN while varying the E2E distance and minimum required fidelity for third user under different turbulence effects.}\vspace{-0.3cm}
    \label{fig_RIS_heatmaps}
    \vspace{-0.cm}
\end{figure}

Based on the derived expression for the E2E entangled state in Proposition \ref{e2e_state_proposition}, we observe that the turbulence strength and the E2E distance are the key factors that directly affect the E2E fidelity. Hence, in Fig. \ref{fig_RIS_heatmaps}, we analyze the impact of the E2E distance, along with the minimum required fidelity on the achieved E2E sum rate in a three-user QN under moderate and strong turbulence effects. Particularly, we consider Mote Carlo simulations for the location of two users with a minimum required fidelity of 0.5. Simultaneously, the location of the third user is varied in the tail of the considered truncated normal distribution of user locations such that the distance between the QBS and user 3 is between 400\,m and 450\,m, and its required minimum fidelity, $f_{\text{min},3}$ is varied between 0.5 and 0.7, as shown in Fig. \ref{fig_RIS_heatmaps}. Under moderate turbulence, we observe from Fig. \ref{fig_RIS_heatmaps}(a) that the impact of increasing the E2E distance on the resulting E2E fidelity is more severe than the impact of increasing the minimum required fidelity, in moderate quantum applications. This becomes increasingly evident under stronger turbulence, as shown in Fig. \ref{fig_RIS_heatmaps}(b). In particular, we observe from Fig. \ref{fig_RIS_heatmaps}(b) that, when the QN has users far away from the QBS and under strong turbulence effects, the QBS could become unable to serve all users with entangled qubits that satisfy their various fairness, rate, and fidelity requirements. For instance, when user 3 has the lowest minimum required fidelity of 0.5 under strong turbulence, the largest distance of that user from the QBS that our framework can tolerate is around 430\,m. One solution to overcome this bottleneck is to use entanglement generation sources with higher fidelity, and to adopt techniques like multiplexing of single-photon sources to enhance the initial entanglement generation attempt rate, which would ensure the QBS has enough resources to overcome the increased losses and noise due to strong turbulence.   

\begin{figure}
    \centering
 \includegraphics[width=0.95\linewidth]{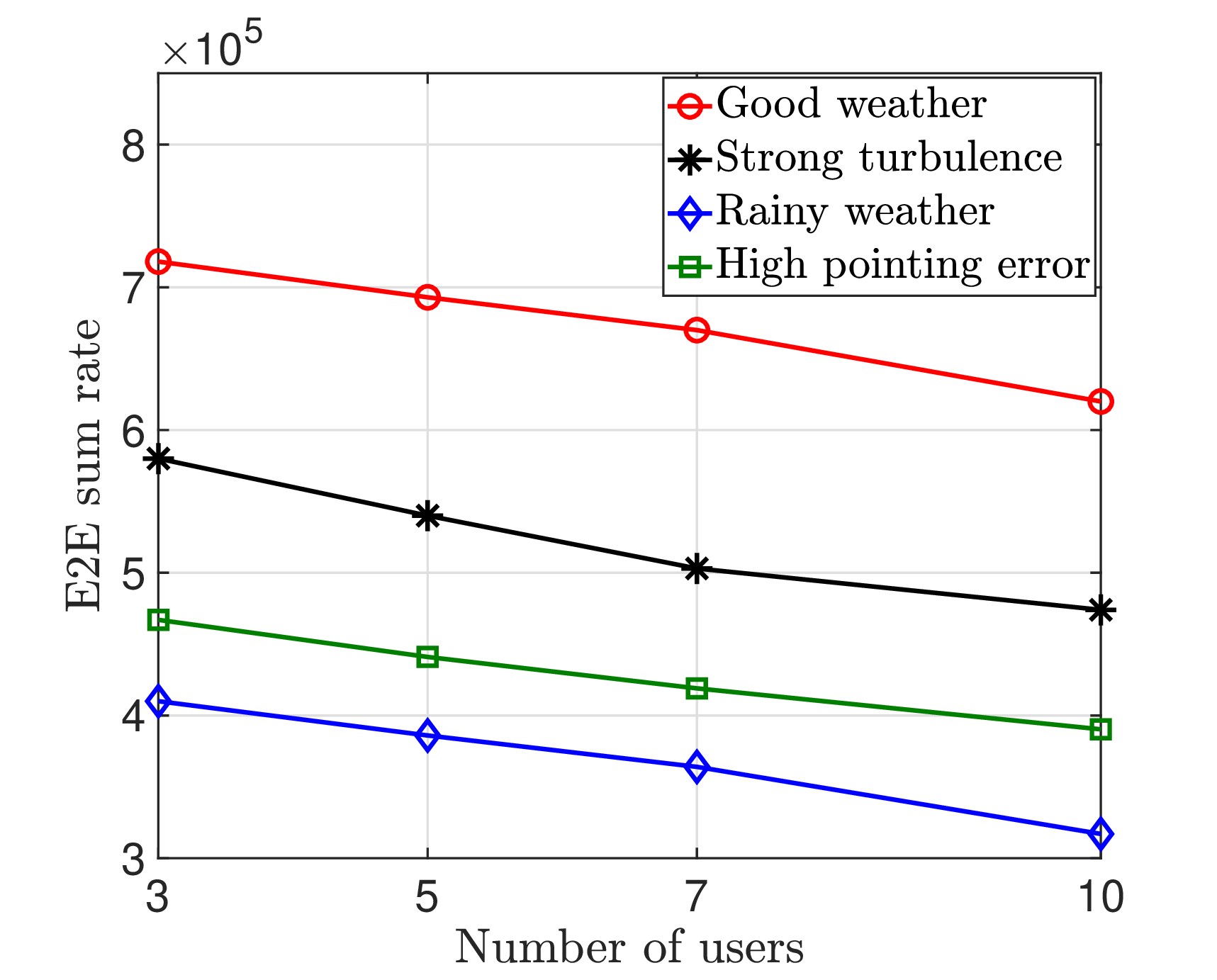}\vspace{-0.1cm}
    \caption{Achieved E2E sum rate for a varying number of users under different environmental conditions.}\vspace{-0.5cm}
    \label{fig_results_scalability}
\end{figure}

Finally, Fig. \ref{fig_results_scalability} compares the obtained E2E sum rate in the default setup (sunny weather, moderate turbulence, and low pointing error) against the performance obtained under: 1) rainy weather ($\hat{\kappa}=6.27$dB/km), 2) strong turbulence ($C_n^2 = 1\times10^{-13}$m$^{-2/3}$), and 3) high pointing error ($\sigma_\theta = 3$\,mrad and $\sigma_\phi = 1$\,mrad). From Fig. \ref{fig_results_scalability}, we observe that our proposed framework can successfully serve a high number of QN users while satisfying their different minimum fidelity, rate, and fairness constraints. Here, we note that the E2E sum rate decreases as the number of users increases. This is due to the fairness constraint, which limits the flexibility in delivering E2E rates to users with better channel conditions in order to satisfy the minimum requirements for all QN users. Furthermore, we observe from Fig. \ref{fig_results_scalability} that the weather condition (e.g., sunny vs. rainy) has the most significant impact among the different environmental effects on the overall network performance. Particularly, changing the weather from being sunny to being rainy results in above 45\% reduction in the achieved E2E sum rate, while having a high pointing error results in around 33\% reduction, and a strong turbulence results in around 18\% reduction in the achieved E2E sum rate. 


\vspace{-0.4cm}
\section{Conclusion}\label{sec_conclusion}\vspace{-0.2cm}
In this paper, we have proposed the use of an RIS in an FSO-based star-shaped QN that is replete with blockages under various environmental conditions. We have analyzed the various losses experienced by the quantum signals over the FSO channel and proposed a novel model for the resulting quantum noise and its effect on the quality of quantum signals. We have then developed a novel framework that jointly optimizes the RIS placement and the initial EGR allocation among the different QN users. This joint optimization is solved while ensuring a minimum fairness level between the users and satisfying their heterogeneous QoS requirements on minimum rate and fidelity levels. Simulation results showed that our framework outperforms existing classical resource allocation frameworks that also fail to satisfy the minimum fidelity requirements. Additionally, we have verified the scalability of our framework, showed its greater sensitivity to weather conditions over pointing errors and turbulence, and found that a user's E2E distance more strongly influences the E2E rate than minimum required fidelity does.
\vspace{-0.4cm}
\appendices
\section{Proof of Theorem 1}\label{appendix_proof_Theorem_1}
To prove Theorem \ref{thrm_prob_success}, we use the PDFs for the pointing errors and the atmospheric turbulence to find the probability of success. For this purpose, we first find the PDF for the combined attenuation, i.e., $h_i, \forall i\in\mathcal{N}$, as follows:

\vspace{-0.5cm}
\begin{align}
    f_{h_i}(h_i) = \int f_{h_i|h_{i}^\text{a}}\left(h_i|h_{i}^\text{a}\right) f_{h_{i}^\text{a}}\left(h_{i}^\text{a}\right)  dh_{i}^\text{a}, \label{CombPDF}
\end{align}
where $f_{h_i|h_{i}^\text{a}}\left(h_i|h_{i}^\text{a}\right)$ is the conditional probability given $h_{i}^\text{a}$, defined as follows:

\begin{equation}
    \begin{split}
    &f_{h_i|h_{i}^\text{a}}\left(h_i|h_{i}^\text{a}\right) = \frac{1}{h_{i}^\text{a}h_{i}^\text{p}}f_{h_{i}^\text{g}}\left(\frac{h_i}{h_{i}^\text{a}h_{i}^\text{p}}\right),\\
    & = \frac{\vartheta_i}{A_{0,i}^{\vartheta_i}h_{i}^\text{a}h_{i}^\text{p}} \left(\frac{h_i}{h_{i}^\text{a}h_{i}^\text{p}}\right)^{\vartheta_i-1},\\
    &0 \leq h_i \leq A_{0,i}h_{i}^\text{a}h_{i}^\text{p},\quad \forall i \in \mathcal{N}. \label{CondProb}
\end{split}
\end{equation}

By substituting \eqref{TurbPDF} and \eqref{CombPDF} in \eqref{CondProb}, we have:
\begin{align}
    &f_{h_i}(h_{i})={{2\vartheta_i(\alpha_i\beta_i)^{(\alpha_i+\beta_i)/2}}\over{\left(A_{0,i}h_{i}^\text{p}\right)^{\vartheta_i}\Gamma(\alpha_i)\Gamma(\beta_i)}} h_i^{\vartheta_i - 1}\times \notag \\&\int_{\frac{h_i}{A_{0,i}h_{i}^\text{p}}}^\infty (h_{i}^\text{a})^{((\alpha_i+\beta_i)/2)-1-\vartheta_i}K_{\alpha_i-\beta_i}\left(2\sqrt{\alpha_i\beta_i h_{i}^\text{a}}\right) d_{h_{i}^\text{a}}, \label{eq_1}
\end{align}
The integral in \eqref{eq_1} can be numerically computed. Particularly, to get a closed-form expression, we can express $K_v(.)$ in terms of the Meijer's G-function $G_{p,q}^{m,n}$. In particular, $K_v(.)$ can be written as
\begin{align}
    K_{\alpha - \beta}(x) = \frac{1}{2}G_{0,2}^{2,0}   \left[\frac{x^2}{4}{\Bigg\vert}{\begin{matrix} \multicolumn{2}{c}{-} \\ \frac{\alpha - \beta}{2} & \frac{\beta - \alpha}{2}\end{matrix}}\right].
\end{align}
Then, the PDF for the combined channel gain will be:
\begin{align}
    &f_{h_i}(h_i) = {\frac{\vartheta_i\alpha_i\beta_i}{A_{0,i}h_{i}^\text{p}\Gamma(\alpha_i)\Gamma(\beta_i)}} \left(\frac{\alpha_i\beta_ih_i}{A_{0,i}h_{i}^\text{p}}\right)^{((\alpha_i+\beta_i)/2) - 1} \times \notag \\ &G_{1,3}^{3,0} \left[\frac{\alpha_i\beta_i}{A_{0,i}h_{i}^\text{p}}h_i {\Bigg\vert}{\begin{matrix} \multicolumn{3}{c}{1 - \frac{\alpha_i + \beta_i}{2} + \vartheta_i} \\ -\frac{\alpha_i + \beta_i}{2} + \vartheta_i, & \frac{\alpha_i - \beta_i}{2}, &  \frac{\beta_i - \alpha_i}{2}\end{matrix}}\right], \label{simp}
\end{align}
\eqref{simp} can be simplified to
\begin{align}
    &f_{h_i}(h_i) = {\frac{\vartheta_i\alpha_i\beta_i}{A_{0,i}h_{i}^\text{p}\Gamma(\alpha_i)\Gamma(\beta_i)}} \times \notag \\ &G_{1,3}^{3,0} \left[\frac{\alpha_i\beta_i}{A_{0,i}h_{i}^\text{p}}h_i {\Bigg\vert}{\begin{matrix} \multicolumn{3}{c}{ \vartheta_i} \\  \vartheta_i - 1, & \alpha_i - 1, &  \beta_i - 1 \end{matrix}}\right].
\end{align}
After obtaining the PDF for the combined channel gain, we can derive the probability of success, which can be defined as
\begin{align}
    P_{\mathrm{Succ}}(\zeta_{\mathrm{th}}) &= P\left(h_i>\zeta_{\mathrm{th}}\right), \notag \\
    & = 1 - F_{h_i}(\zeta_{\mathrm{th}}),\notag \\
    & \overset{(a)}{=} 1 - \int_0^{\zeta_{\mathrm{th}}} f_{h_i}(h_i) dh_i,
\end{align}
where $F_h(.)$ is the CDF of random variable $h$, and $\zeta_{\mathrm{th}}$ is the predefined channel gain threshold. Note that the integral in (a) can be solved by expressing the $K_v(.)$ in terms of the Meijer’s $G$-function, which leads to the following:
    \begin{equation}\small
    \begin{split}
        &P_{\mathrm{succ},i}(\textit{\textbf{l}}_r) = 1-\Bigg(\frac{\vartheta_i(\textit{\textbf{l}}_r)}{\Gamma(\alpha_i(\textit{\textbf{l}}_r))\Gamma(\beta_i(\textit{\textbf{l}}_r))} \\
        & \times G_{2,4}^{3,1} \left[\frac{\alpha_i(\textit{\textbf{l}}_r)\beta_i(\textit{\textbf{l}}_r) \chi_{\mathrm{th}}}{A_{0,i}(\textit{\textbf{l}}_r)h_{i}^\text{p}(\textit{\textbf{l}}_r)} {\Bigg\vert} \begin{matrix} 
            \multicolumn{2}{c}{1,} & \multicolumn{2}{c}{\vartheta_i(\textit{\textbf{l}}_r) +1} \\ 
            \vartheta_i(\textit{\textbf{l}}_r), & \alpha_i(\textit{\textbf{l}}_r), & \beta_i(\textit{\textbf{l}}_r),&0
            \end{matrix}\right]\Bigg),
    \end{split}
    \end{equation}
where $\chi_{\mathrm{th}} = \frac{\zeta_{\mathrm{th}}}{\varsigma \eta}$. This completes the proof.

\vspace{-0.3cm}
\section{Proof of Proposition 1}\label{appendix_lemma_final_state_derivation}
First, we consider the impact of applying a depolarizing noise channel on the matter qubit. Consider the action of a depolarizing channel $\boldsymbol{\Lambda}_{p_1}$ on a general single-qubit state $\boldsymbol{\rho}$:
\begin{equation}
\boldsymbol{\Lambda}_{p_1}(\boldsymbol{\rho}) = (1-p_1)\boldsymbol{\rho} + p_1\frac{\boldsymbol{I}}{2},
\end{equation}
where $\boldsymbol{I}$ is the identity matrix for a single qubit. This essentially means that with probability $p_1$, the state is replaced by the completely mixed state $\frac{\boldsymbol{I}}{2}$, and with probability $1-p_1$, the state, $\boldsymbol{\rho}$, remains the same.

Now, considering a Bell-diagonal quantum state $\boldsymbol{\rho}_{\mathrm{BD},i}$ generated by the QBS, which is a two-qubit state, if we apply the depolarizing noise to the first qubit of the Bell-diagonal state $\boldsymbol{\rho}_{\mathrm{BD},i}$ for user $i\in\mathcal{N}$, the final resulting state $\boldsymbol{\rho}_{\mathrm{BD},i}'$ will be:
\begin{equation}\label{eq_1_noise}
\boldsymbol{\rho}_{\mathrm{BD},i}' = (1-p_{1,i})(\boldsymbol{\rho}_{\mathrm{BD},i}) + p_{1,i}(\boldsymbol{I} \otimes \boldsymbol{\rho}_{2})/2,
\end{equation}
where $\boldsymbol{\rho}_{2}$ is the reduced density matrix of the second qubit, obtained by tracing out the first qubit from $\boldsymbol{\rho}_{\mathrm{BD},i}$.

In the case of a Bell-diagonal state, $\boldsymbol{\rho}_{2}$ is a completely mixed state $\left(\frac{\boldsymbol{I}}{2}\right)$, and the state $\boldsymbol{\rho}'$ can be simplified as:
\begin{equation}\label{eq_rho_prime}
\boldsymbol{\rho}_{\mathrm{BD},i}' = \left(1-p_1\right)\boldsymbol{\rho}_{\mathrm{BD},i} + p_{1,i}\left(\frac{\boldsymbol{I} \otimes \boldsymbol{I}}{4}\right).
\end{equation}

For a single qubit, the identity operator $\boldsymbol{I}$ can be represented as:
\begin{equation}
\boldsymbol{I} = |0\rangle \langle 0| + |1\rangle \langle 1|.
\end{equation}

Accordingly, 
\begin{equation}
\boldsymbol{I} \otimes \boldsymbol{I} = |00\rangle \langle 00| + |01\rangle \langle 01| + |10\rangle \langle 10| + |11\rangle \langle 11|.
\end{equation}

Given that: 
\begin{equation}
\begin{split}
    \boldsymbol{\Phi}_{00} &= \frac{1}{2}(\ket{00}\bra{00} + \ket{00}\bra{11} + \ket{11}\bra{00} + \ket{11}\bra{11})\\
    \boldsymbol{\Phi}_{01} &= \frac{1}{2}(\ket{00}\bra{00} - \ket{00}\bra{11} - \ket{11}\bra{00} + \ket{11}\bra{11})\\
    \boldsymbol{\Phi}_{10} &= \frac{1}{2}(\ket{01}\bra{01} + \ket{01}\bra{10} + \ket{10}\bra{01} + \ket{10}\bra{10})\\
    \boldsymbol{\Phi}_{11} &= \frac{1}{2}(\ket{01}\bra{01} - \ket{01}\bra{10} - \ket{10}\bra{01} + \ket{10}\bra{10}),
\end{split}
\end{equation}
we can represent $\boldsymbol{I} \otimes \boldsymbol{I}$ in terms of the Bell states as follows (normalized): 
\begin{equation}
(\boldsymbol{I} \otimes \boldsymbol{I})/4 = (\boldsymbol{\Phi}_{00} + \boldsymbol{\Phi}_{01} + \boldsymbol{\Phi}_{10} + \boldsymbol{\Phi}_{11})/4,
\end{equation}

Given that $p_{1,i}\left(\textit{\textbf{l}}_r\right) = \left(1-e^{-\frac{t_i\left(\textit{\textbf{l}}_r\right)}{T}}\right)$, and $\boldsymbol{\rho}_{\mathrm{BD},i} = \lambda_{00,i}\boldsymbol{\Phi}_{00} + \lambda_{01,i}\boldsymbol{\Phi}_{01} + \lambda_{10,i}\boldsymbol{\Phi}_{10} + \lambda_{11,i}\boldsymbol{\Phi}_{11}$, 
by substituting into \eqref{eq_rho_prime}, and simplifying, we get:
\begin{equation}\small
\begin{split}
\boldsymbol{\rho}_{\mathrm{BD},i}'(\textit{\textbf{l}}_r)& = e^{-\frac{t_i(\textit{\textbf{l}}_r)}{T}} \times \left(\lambda_{00,i}\boldsymbol{\Phi}_{00} + \lambda_{01,i}\boldsymbol{\Phi}_{01}+ \lambda_{10,i}\boldsymbol{\Phi}_{10} \right.\\ &+ \left. \lambda_{11,i}\boldsymbol{\Phi}_{11}\right) + \left(\frac{1}{4}\left(1-e^{-\frac{t_i\left(\textit{\textbf{l}}_r\right)}{T}}\right)\right) \times (\boldsymbol{\Phi}_{00}\\ 
& + \boldsymbol{\Phi}_{01} + \boldsymbol{\Phi}_{10} + \boldsymbol{\Phi}_{11}),
\end{split}
\end{equation}
which can be rewritten concisely as:
\begin{equation}\label{simplified_eq_after_first_channel}
\begin{split}
\boldsymbol{\rho}_{\mathrm{BD},i}'(\textit{\textbf{l}}_r) &= F_{00}(\textit{\textbf{l}}_r)\boldsymbol{\Phi}_{00} + F_{01}(\textit{\textbf{l}}_r)\boldsymbol{\Phi}_{01}\\ &+ F_{10}(\textit{\textbf{l}}_r)\boldsymbol{\Phi}_{10} + F_{11}(\textit{\textbf{l}}_r)\boldsymbol{\Phi}_{11},
\end{split}
\end{equation}
where $F_{jk}(\textit{\textbf{l}}_r) = \left(\frac{1}{4} + \left(\lambda_{jk,i}-\frac{1}{4}\right)e^{-\frac{t_i\left(\textit{\textbf{l}}_r\right)}{T}} \right)$, and $F_{00}+F_{01}+F_{10}+F_{11} = 1$.

Note that the resulting state $\boldsymbol{\rho}_{\mathrm{BD},i}'$ is a valid density matrix, maintaining complete positivity and trace preservation, fundamental for the consistency of quantum mechanics.

The phase damping channel effectively flips the phase of the state with a certain probability $p_{2,i}$. Hence, we apply the phase damping channel only to the second qubit of the state $\boldsymbol{\rho}_{\mathrm{BD},i}'$, so the final E2E state $\boldsymbol{\rho}_{\mathrm{BD},i}''$ will be:
\begin{equation}
\begin{split}
\boldsymbol{\rho}_{\mathrm{BD},i}''(\textit{\textbf{l}}_r) &= \left(1-p_{2,i}\left(\textit{\textbf{l}}_r\right)\right)\left(\boldsymbol{\rho}_{\mathrm{BD},i}'\right)\\ &+ p_{2,i}\left(\textit{\textbf{l}}_r\right)\left(\boldsymbol{I} \otimes \boldsymbol{\sigma}_Z\right)\boldsymbol{\rho}_{\mathrm{BD},i}'\left(\boldsymbol{I} \otimes \boldsymbol{\sigma}_Z\right).
\end{split}
\end{equation}

First, we observe that $\boldsymbol{I} \otimes \boldsymbol{\sigma}_Z = |00\rangle\langle00| - |01\rangle\langle01| + |10\rangle\langle10| - |11\rangle\langle11|$. As such, We conclude that:   
\begin{equation}
    \begin{split}
        (\boldsymbol{I} \otimes \boldsymbol{\sigma}_Z)\boldsymbol{\Phi}_{00}(\boldsymbol{I} \otimes \boldsymbol{\sigma}_Z) &= \boldsymbol{\Phi}_{01}\\
        (\boldsymbol{I} \otimes \boldsymbol{\sigma}_Z)\boldsymbol{\Phi}_{01}(\boldsymbol{I} \otimes \boldsymbol{\sigma}_Z) &= \boldsymbol{\Phi}_{00}\\
        (\boldsymbol{I} \otimes \boldsymbol{\sigma}_Z)\boldsymbol{\Phi}_{10}(\boldsymbol{I} \otimes \boldsymbol{\sigma}_Z) &= \boldsymbol{\Phi}_{11}\\
        (\boldsymbol{I} \otimes \boldsymbol{\sigma}_Z)\boldsymbol{\Phi}_{11}(\boldsymbol{I} \otimes \boldsymbol{\sigma}_Z) &= \boldsymbol{\Phi}_{10}.
    \end{split}
\end{equation}
Accordingly, we find that:
\begin{equation}\label{second_part_after_second_noise_channel}
\begin{split}
(\boldsymbol{I} \otimes \boldsymbol{\sigma}_Z)\boldsymbol{\rho}_{\mathrm{BD},i}'(\boldsymbol{I} \otimes \boldsymbol{\sigma}_Z) &= F_{01}(\textit{\textbf{l}}_r)\boldsymbol{\Phi}_{00} + F_{00}(\textit{\textbf{l}}_r)\boldsymbol{\Phi}_{01}\\ &+ F_{11}(\textit{\textbf{l}}_r)\boldsymbol{\Phi}_{10} + F_{10}(\textit{\textbf{l}}_r)\boldsymbol{\Phi}_{11},
\end{split}
\end{equation}
which corresponds to the following expression of the E2E entangled state:
\begin{equation}\label{eq_final_state_lemma_proof}
\begin{split}
    \boldsymbol{\rho}_{i}''(\textit{\textbf{l}}_r) &= \Big[(1-p_{2,i}(\textit{\textbf{l}}_r))F_{00}(\textit{\textbf{l}}_r) + p_{2,i}(\textit{\textbf{l}}_r)F_{01}(\textit{\textbf{l}}_r))\Big]\boldsymbol{\Phi}_{00}\\ &+ \Big[(1-p_{2,i}(\textit{\textbf{l}}_r))F_{01}(\textit{\textbf{l}}_r) + p_{2,i}(\textit{\textbf{l}}_r)F_{00}(\textit{\textbf{l}}_r))\Big]\boldsymbol{\Phi}_{01}\\ 
    &+ \Big[(1-p_{2,i}(\textit{\textbf{l}}_r))F_{10}(\textit{\textbf{l}}_r) + p_{2,i}(\textit{\textbf{l}}_r)F_{11}(\textit{\textbf{l}}_r))\Big]\boldsymbol{\Phi}_{10}\\ &+ \Big[(1-p_{2,i}(\textit{\textbf{l}}_r))F_{11}(\textit{\textbf{l}}_r) + p_{2,i}(\textit{\textbf{l}}_r)F_{10}(\textit{\textbf{l}}_r))\Big]\boldsymbol{\Phi}_{11},
\end{split}
\end{equation}
which is another Bell-diagonal state that is completely-positive and trace-preserving. 

By comparison with the general formula of Bell-diagonal quantum states,  $\boldsymbol{\rho}_{i}''(\textit{\textbf{l}}_r) = \lambda_{00,i}''(\textit{\textbf{l}}_r)\boldsymbol{\Phi}_{00} + \lambda_{01,i}''(\textit{\textbf{l}}_r)\boldsymbol{\Phi}_{01} + \lambda_{10,i}''(\textit{\textbf{l}}_r)\boldsymbol{\Phi}_{10} + \lambda_{11,i}''(\textit{\textbf{l}}_r)\boldsymbol{\Phi}_{11}$, we conclude that:
\begin{equation}
    \begin{split}   \lambda_{00,i}''(\textit{\textbf{l}}_r) &= \Big[(1-p_{2,i}(\textit{\textbf{l}}_r))F_{00}(\textit{\textbf{l}}_r) + p_{2,i}(\textit{\textbf{l}}_r)F_{01}(\textit{\textbf{l}}_r))\Big]\\
         \lambda_{01,i}''(\textit{\textbf{l}}_r) &= \Big[(1-p_{2,i}(\textit{\textbf{l}}_r))F_{01}(\textit{\textbf{l}}_r) + p_{2,i}(\textit{\textbf{l}}_r)F_{00}(\textit{\textbf{l}}_r))\Big]\\
         \lambda_{10,i}''(\textit{\textbf{l}}_r) &= \Big[(1-p_{2,i}(\textit{\textbf{l}}_r))F_{10}(\textit{\textbf{l}}_r) + p_{2,i}(\textit{\textbf{l}}_r)F_{11}(\textit{\textbf{l}}_r))\Big]\\
         \lambda_{11,i}''(\textit{\textbf{l}}_r) &= \Big[(1-p_{2,i}(\textit{\textbf{l}}_r))F_{11}(\textit{\textbf{l}}_r) + p_{2,i}(\textit{\textbf{l}}_r)F_{10}(\textit{\textbf{l}}_r))\Big],
    \end{split}
\end{equation}
where $p_{2,i}\left(\textit{\textbf{l}}_r\right) = \text{erf}\left(\sigma_R^2\left(\textit{\textbf{l}}_r\right)\right)$, and $F_{jk}\left(\textit{\textbf{l}}_r\right) = \left(\frac{1}{4} + (\lambda_{jk,i}-\frac{1}{4})e^{-\frac{t_i\left(\mathbf{l}_r\right)}{T}} \right)$. This is the closed-form expression of the E2E state for user $i\in\mathcal{N}$, which is necessary to quantify its fidelity that must satisfy a minimum threshold to be useful in quantum applications.

\vspace{-0.4cm}
\bibliographystyle{IEEEtran}
\bibliography{references,quantum_noise}


 




\end{document}